\documentclass[11pt]{article}

\usepackage[margin=1in]{geometry}
\usepackage[utf8]{inputenc}
\usepackage[T1]{fontenc}
\usepackage{url}
\usepackage{booktabs}
\usepackage{amsfonts}
\usepackage{nicefrac}
\usepackage{microtype}
\usepackage{setspace}
\usepackage[compress,sort]{natbib}
\usepackage{bbm}
\usepackage[table]{xcolor}
\usepackage{amsthm}
\usepackage{amsmath}
\usepackage{amssymb}
\usepackage{mathtools}
\usepackage{enumitem}
\usepackage{soul}
\usepackage{tabu}
\usepackage{xcolor}
\usepackage{comment}
\usepackage{bbold}
\usepackage{subcaption}
\usepackage{tikz-cd}
\usepackage[hidelinks]{hyperref}
\usetikzlibrary{positioning}
\usetikzlibrary{calc}
\usetikzlibrary{arrows}
\usetikzlibrary{decorations.markings}
\usetikzlibrary{arrows.meta}
\usetikzlibrary{bayesnet}
\usetikzlibrary{patterns}
\usetikzlibrary{patterns.meta}
\usetikzlibrary{decorations.pathreplacing}
\usepackage{float}
\tikzset{strike thru arrow/.style={
  decoration={markings, mark=at position 0.5 with {
    \draw [blue, thick,-]
      ++ (-0.15cm, -0.15cm)
      -- ( 0.15cm,  0.15cm);}
  },
  postaction={decorate},
}}
\definecolor{blue1}{HTML}{2171b5}
\definecolor{blue2}{HTML}{6baed6}
\definecolor{blue3}{HTML}{bdd7e7}
\definecolor{blue4}{HTML}{eff3ff}

\theoremstyle{plain}
\newtheorem{thm}{Theorem}
\newtheorem{prop}[thm]{Proposition}

\newtheorem{cor}[thm]{Corollary}

\theoremstyle{definition}
\newtheorem{defn}[thm]{Definition}

\newtheorem{cas}{Case}

\usepackage{chngcntr}
\usepackage{apptools}
\AtAppendix{\counterwithin{thm}{section}}

\theoremstyle{remark}

\newcommand{\nc}{\newcommand}
\nc{\dmo}{\DeclareMathOperator}
\newcommand{\rnc}{\renewcommand}

\nc{\B}[1]{\mathbb{#1}}
\nc{\C}[1]{\mathcal{#1}}
\nc{\Sc}[1]{\mathscr{#1}}
\nc{\F}[1]{\mathfrak{#1}}
\nc{\cdes}[1][]{\sate_{M}}
\nc{\cdeshat}[1][]{\widehat{\sate}_{M}}
\nc{\tes}[1][]{\te_{\Ob #1}}
\nc{\ates}[1][]{\ate_{\Ob #1}}
\nc{\rsqy}{R^2_{Y \sim U \mid Z, X}}
\nc{\rsqz}{R^2_{Z \sim U \mid X}}
\nc{\wt}{\widetilde}
\nc{\wh}{\widehat}
\nc{\la}{\langle}
\nc{\ra}{\rangle}
\nc{\fvert}{\left\vert\vphantom{\frac11}\right.}

\nc{\indep}{\protect\mathpalette{\protect\independenT}{\perp}}
\def\independenT#1#2{\mathrel{\rlap{\(#1#2\)}\mkern2mu{#1#2}}}

\rnc{\epsilon}{\varepsilon}

\dmo{\cde}{\textsc{cde}}
\dmo{\te}{\textsc{te}}
\dmo{\ate}{\textsc{ate}}
\dmo{\cate}{\textsc{cate}}
\dmo{\sate}{\textsc{sate}}
\dmo{\var}{\textsc{Var}}
\dmo{\bnm}{\textsc{Binom}}
\dmo{\EE}{\mathbb{E}}
\dmo{\Ob}{Ob}
\dmo{\bern}{\textsc{Bern}}
\dmo{\unif}{\textsc{Unif}}
\dmo{\se}{\textsc{se}}
\dmo{\logit}{logit}
\dmo{\Do}{\textit{do}}
\dmo{\NA}{NA}

\title{%
  A Causal Framework for Observational Studies of Discrimination\footnote{%
    We thank Alex Chohlas-Wood, Avi Feller, Andrew Gelman, Zhiyuan ``Jerry''
    Lin, Julian Nyarko, Brendan O'Flaherty, Elizabeth Ogburn, Jos{\'e} Luis
    Montiel Olea, Steven Raphael, James Robins, Rajiv Sethi, Amy Shoemaker, and
    Ilya Shpitser for helpful conversations. Code to replicate our analysis is
    available online at: \url{https://github.com/stanford-policylab/gcbsgh-rep}.
}}
\author{%
  Johann Gaebler\\
  Stanford University\\
  \and
  William Cai\\
  Stanford University\\
  \and
  Guillaume Basse \\
  Stanford University
  \and
  Ravi Shroff \\
  New York University
  \and
  Sharad Goel\\
  Harvard University\\
  \and
  Jennifer Hill \\
  New York University \\
}
\date{}

\begin{document}

\maketitle
\thispagestyle{empty}

\begin{abstract}
  In studies of discrimination, researchers often seek to estimate a causal
  effect of race or gender on outcomes. For example, in the criminal justice
  context, one might ask whether arrested individuals would have been
  subsequently charged or convicted had they been a different race. It has long
  been known that such counterfactual questions face measurement challenges
  related to omitted-variable bias, and conceptual challenges related to the
  definition of causal estimands for largely immutable characteristics. Another
  concern, which has been the subject of recent debates, is post-treatment bias:
  many studies of discrimination condition on apparently intermediate outcomes,
  like being arrested, that themselves may be the product of discrimination,
  potentially corrupting statistical estimates. There is, however, reason to be
  optimistic. By carefully defining the estimand---and by considering the
  precise timing of events---we show that a primary causal quantity of interest
  in discrimination studies can be estimated under an ignorability condition
  that may hold approximately in some observational settings. We illustrate
  these ideas by analyzing both simulated data and the charging decisions of a
  prosecutor's office in a large county in the United States.
\end{abstract}

\newpage

To assess the role of race or gender in decision making, researchers often
examine disparities between groups after adjusting for relevant factors. For
example, to measure racial discrimination in lending decisions, one might
estimate race-specific approval rates after adjusting for creditworthiness,
typically via a regression model. This simple statistical strategy---sometimes
called benchmark analysis---has been used to study discrimination in a wide
variety of domains, including banking~\citep{munnell1996mortgage},
employment~\citep{berg2002measuring}, education~\citep{baum2005gender},
healthcare~\citep{balsa2005testing}, housing~\citep{greenberg2016discrimination,
edelman2014digital}, and criminal justice~\citep{rehavi2014racial,
fryer2019empirical, gelman2007analysis, ayres2002, macdonald2017analysis}.

The results of benchmark analyses are often framed in causal terms (e.g., as an
effect of race on outcomes), but it is well understood that such an approach
suffers from at least three significant statistical challenges when used to
estimate causal quantities. First, at a conceptual level, it is unclear how
best to rigorously define causal estimands of interest when the treatment is
race, gender, or another largely immutable trait. Second, estimates can be
plagued by omitted-variable bias if one does not appropriately adjust for all
relevant covariates. Third---and the focus of our paper---there are worries
that estimates are corrupted by post-treatment bias when one adjusts for
covariates or restricts to samples of individuals determined downstream from
race, gender, or another such treatment variable. This concern, in particular,
has raised doubts about the reliability of the literature on police
discrimination, where many studies rely on administrative stop records, and
hence implicitly condition on officers stopping an individual, an action that
itself is likely discriminatory~\citep{knox-2019, heckman2020comment}.

Here we present a causal framework for conceptualizing and estimating a measure
of discrimination that is suitable for many applied problems. Our framing
specifically addresses concerns about post-treatment bias. To do so, we first
define a causal quantity---the second-stage sample average treatment effect, or
\(\cdes\)---which closely maps to the legal notion of disparate treatment. For
this estimand, by carefully considering the timing of events, we show that
treatment assignment conceptually occurs after selection into the sample of
interest. We then introduce the notion of subset ignorability, show that this
condition formally justifies the use of benchmark analysis to estimate the
\(\cdes\), and discuss settings in which it is likely to hold approximately. We
illustrate these ideas by analyzing synthetic data, as well as a detailed
dataset of prosecutorial charging decisions for approximately 20,000 felony
arrests in a major U.S.\ county. By developing this statistical foundation, we
hope to place discrimination studies on more solid theoretical footing.

\section{A Motivating Example}
\label{sec:motex}

Consider the problem of measuring racial discrimination in prosecutorial
charging decisions. After an individual has been arrested, prosecutors in the
district attorney's office read the arresting officer's incident report and then
decide whether or not to press charges. For simplicity, suppose prosecutors
only have access to the incident report---and to no other information---when
making their decisions. We allow for the possibility that the arrest decision
that preceded the charging decision may have suffered from racial discrimination
in complex ways that cannot be inferred from the incident reports themselves.
Finally, suppose that a researcher has access to these incident reports for
arrested individuals, but, importantly, not to any data on individuals that
officers considered but ultimately decided against arresting. What, if
anything, might one hope to discover about racial discrimination in charging
decisions in light of the fact that the people about whom the prosecutor makes
charging decisions have been selected---that is, arrested---not randomly, but
rather in ways that likely depended on their race?

The first challenge is to rigorously define causal estimands of interest. The
inherent difficulty is captured by the statistical refrain ``no causation
without manipulation''~\citep{holland1986statistics}, since it is often unclear
what it means to alter attributes like race and gender~\citep{sekhonneyman}.
One common maneuver is to instead consider the causal effect of \emph{perceived}
attributes (e.g., perceived race or perceived gender), which ostensibly can be
manipulated---for example, by changing the name listed on an employment
application~\citep{bertrand2004emily}, or by masking an individual's
appearance~\citep{goldin2000orchestrating, grogger2006testing, OPP}. In our
case, one might imagine a hypothetical experiment in which explicit mentions of
race in the incident report are altered (e.g., replacing ``white'' with
``Black''). The causal effect is then, by definition, the difference in
charging rates between those cases in which arrested individuals were randomly
described (and hence may be perceived) as ``Black'' and those in which they were
randomly described as ``white.'' This conceptualization of discrimination
conforms to one common causal understanding of discrimination used, for example,
in audit studies. This framing also maps closely to the legal notion of
disparate treatment, a form of discrimination in which actions are motivated by
animus or otherwise discriminatory intent~\citep{goel2017combatting}.

While researchers have carried out such audit studies---including in the case of
prosecutorial charging decisions
\citep{chohlaswood2020blind,robertson_race_2019}\footnote{%
  There are some differences between the idealized audit study described above
  and these two experiments. \citeauthor{chohlaswood2020blind} conduct a
  quasi-random field trial in which they mask---but do not switch---the stated
  race of individuals in police narratives used to make actual charging
  decisions. \citeauthor{robertson_race_2019} survey prosecutors in a randomized
  lab experiment and ask them, hypothetically, what their charging decision
  would be based on fact patterns in which the race of the suspect is
  manipulated. Although neither of these studies maps exactly to the
  hypothetical experiment motivating our estimand, both demonstrate the
  feasibility of conducting such an experiment.
}---it is often
infeasible to study important policy questions through randomized experiments.
In the absence of a controlled experiment, one can in theory identify this type
of causal estimand from purely observational data by comparing charging rates
across pairs of cases that are identical in all aspects other than the stated
race of the arrested individual.\footnote{%
  It suffices to compare groups of cases that have the same
  distribution of potential outcomes---even if the cases themselves are not
  identical---a property we formalize in Definition~\ref{def:si} below.
} That strategy, which mimics the key features of the hypothetical randomized
experiment described above, is formally justified when treatment assignment
(i.e., description of race on the incident report, and subsequent perception by
the prosecutor) is \emph{ignorable} given the observed covariates (i.e.,
features of the incident report)~\citep{imbens2015causal}. In practice, though,
this approach may suffer from omitted-variable bias when the full incident
report is not available to researchers, and may suffer from lack of overlap when
suitable matches cannot be found for each case---limitations common to many
observational studies of causal effects. To address these issues, one can
restrict attention to the overlap region and gauge the robustness of estimates
to varying forms and degrees of unmeasured
confounding~\citep{cornfield1959smoking, rr, cinelli2018making}, an approach we
demonstrate below.

Finally, there is the issue of post-treatment bias, especially due to sample
selection. \citet{knox-2019}\ argue that researchers often inadvertently
introduce post-treatment bias in observational studies of discrimination by
subsetting on apparently intermediate outcomes---such as, in our charging
example, being arrested---that themselves may be the product of discrimination.
As a result, the authors caution that causal quantities of interest cannot be
identified by the data in the absence of implausible assumptions, such as lack
of discrimination in the initial arrest decision. In making their argument,
\citeauthor{knox-2019}\ focus on the use of force by police officers in civilian
encounters, but they suggest their formal critique applies more broadly, casting
doubt on a wide range of observational studies of discrimination.

Here we show that such customary subsetting does not pose an insurmountable
threat to discrimination research. To understand why, one must precisely define
the causal estimand, and carefully consider the timing of events. For instance,
in our charging example, there are two relevant treatments, the officer's
perception of race, affecting the officer's arrest decision, and the
prosecutor's perception of race, affecting the prosecutor's charging decision.
The arrest decision is post-treatment relative to the officer's perception of
race but, importantly, it is pre-treatment relative to the prosecutor's
perception of race. Similarly, the features of the incident report---which we
must adjust for in this type of benchmark analysis---are post-treatment relative
to the officer's perception of race but pre-treatment relative to the
prosecutor's perception of race. In such a two-decider situation, as
\citet{greiner2011causal} suggest, it is possible to recover estimates of
discrimination by the second decider (e.g., in the charging decision) even if
there is discrimination by the first decider (e.g., in the arrest decision).

\section{A Measure of Discrimination}
\label{sec:post_treat_bias}

We present a simple two-stage model to characterize discriminatory decision
making in a variety of real-world situations and define our main causal quantity
of interest---the second-stage sample average treatment effect, or
\(\cdes\)---within this general framework. In the context of our motivating
example, the \(\cdes\) corresponds to the quantity that would be measured in the
hypothetical audit study of prosecutorial decisions described in
Section~\ref{sec:motex}. A central aim of this paper is to formalize technical
assumptions that allow one to statistically identify discrimination---more
precisely, disparate treatment---in the second stage (e.g., in prosecutorial
charging decisions) when data are only available for individuals who made it
past the first stage (e.g., those who were arrested). Importantly, our
formalization accommodates scenarios in which first-stage decisions may
themselves be discriminatory.

In the first stage, we assume each individual in some population is subject to a
binary decision \(M\), such as an offer of employment, admission to college, or
law enforcement action. Those who receive a ``positive'' first-stage decision
(e.g., those who are arrested) proceed to a second stage, where another binary
decision \(Y\) is made. In our running example, the case of each arrested
individual is reviewed in the second stage by a prosecutor who may or may not
choose to press charges. Those who are not arrested do not have a case that
requires review by a prosecutor and, indeed, there may be no administrative
record of those individuals.

When considering racial discrimination in decisions involving Black and white
individuals, our primary quantity of interest is the second-stage sample average
treatment effect, \(\EE[Y(b) - Y(w)]\), where \(Y(z)\) indicates the potential
second-stage decision and the expectation is taken over individuals reaching the
second stage. Here, we imagine that the perception of race is counterfactually
determined after the first-stage decision but before the second-stage decision
(e.g., after arrest but before charging, perhaps by altering the description of
race on the incident report viewed by a prosecutor). The second-stage sample
average treatment effect thus captures discrimination in the second-stage
decision among those who made it past the first stage (e.g., discrimination in
charging decisions among those who were arrested). This estimand maps onto a
common understanding of disparate treatment in second-stage decisions, including
in our charging example.

\subsection{A formal model of discrimination}
\label{ssec:discrimination_model}

We now formalize the above discussion to explicitly include decisions made at
both the first and second stages. For ease of interpretation, we follow
\citet{greiner2011causal} and motivate our statistical model by considering
settings where there are two deciders (e.g., an officer and a prosecutor) whose
perceptions of race---or gender, or another trait---can in theory be
independently altered  prior to their decisions. There are, however, examples
in which one can plausibly intervene twice even when a single decider makes both
decisions. For instance, an officer may decide to stop a motorist based in part
on a brief impression of the motorist's skin tone as they drive
past~\citep{grogger2006testing, OPP}. This visual impression of race could
subsequently be altered if the motorist presents a driver's license bearing a
name characteristic of another race group, or speaks a dialect of English at
odds with the officer's expectation. It thus may be possible to apply our
framework whether one imagines there are two deciders or a single one.

We begin by denoting the race of an individual as perceived by the first decider
at the first stage by \(D \in \{w,b\}\), where, for simplicity, we consider a
population consisting of only white and Black individuals. We focus on racial
discrimination for concreteness, but similar considerations apply to
discrimination based on other attributes, such as gender. Assuming that there is
no interference between units~\citep{imbens2015causal}, we let the binary
variables \(M(w)\) and \(M(b)\) denote the potential first-stage decisions for
an individual (e.g., whether they were arrested), and write \(M = M(D)\) for the
observed first-stage decision. To avoid triviality, we assume throughout that
\(\Pr(M = 1) > 0\).

Next, we let \(Z \in \{w,b\}\) denote the race of an individual as perceived by
the second decider, at the second stage. In our running example, \(Z\) denotes
race as perceived by the prosecutor reviewing that person's file, while \(D\)
denotes race as perceived by the police officer during the encounter. Finally,
we define the second-stage potential outcomes as a function of both the
first-stage outcome \(M\) (e.g., the arrest decision) and the second decider's
perception of race \(Z\). Thus, assuming once again that there is no
interference, the observed second-stage outcome for an individual can be denoted
\(Y = Y(Z, M)\), where we consider four potential second-stage outcomes for each
individual: \(Y(z, m)\), where \(z\in \{w,b\}\) and \(m \in \{0,1\}\). In our
example, only those who were arrested can be charged, and so \(Y(b, 0) = Y(w, 0)
= 0\) for all individuals.\footnote{%
  To avoid imagining values of \(Z\) for individuals not arrested, one could
  also make them ``missing'' by setting \(Z = Z(D, M)\), \(Z(d, 0) = \NA\),
  \(Z(d, 1) = d\), \(Y(\NA, 1) = \NA\), and \(Y(z, 0) = 0\) for \(z \in \{w, b,
  \NA\}\), as we do in the simulation in Section~\ref{sec:ex} below. This does
  not affect any of the mathematical details in what follows.
}

We further allow each individual to have an associated vector of (non-race)
covariates \(X\), representing, for example, their behavior during a police
encounter, their recorded criminal history, or both. We imagine these
covariates are fixed prior to the second-stage treatment (e.g., prior to the
prosecutor's perception of race), since otherwise the key ignorability
assumption in Definition~\ref{def:si} below is unlikely to hold. In practice,
\(X\) is only observed for a subset of the population (e.g., those who were
arrested and hence in the dataset), but we nonetheless define the covariate
vector for all individuals in our population of interest. These covariates are
not necessary to define our causal estimands of interest, but they play an
important role in constructing our statistical estimators.

In this model of discrimination, we have taken care to distinguish between the
(realized) first- and second-stage perceptions of race, \(D\) and \(Z\), because
this helps to clarify the timing of events and the meaning of causal quantities.
Importantly, this makes it clear that we can conceive of \(D\) and \(Z\) as
separately manipulable. At the same time, our focus is observational settings,
in which disagreement between \(Z\) and \(D\) may be realized only rarely, if at
all, in the data we observe. For instance, barring manipulation of the incident
report, it seems unlikely that an arresting officer's perception of race will
frequently differ from a prosecutor's perception. Our simulation in
Section~\ref{sec:ex} thus imposes the further constraint that perceived race is
the same at each stage, though this restriction is not necessary in general.

With this framing, we now formally describe the primary causal estimand we
consider. This quantity, which we call the second-stage sample average treatment
effect (\(\cdes\)) reflects discrimination in the second stage of the
decision-making process outlined above, such as discrimination in the
prosecutor's charging decision.\footnote{%
  The \(\cdes\) is  notationally equivalent to the \(\cde_{M = 1}\) defined in
  \citet{knox-2019}. In our case, however, we have taken care to specify that
  the first parameter in the quantity \(Y(z, m)\) denotes intervening on the
  \emph{second-stage} perception of race. Moreover, the \(\cdes\) is distinct
  from what  \citeauthor{knox-2019}\ call the \(\ate_{M = 1}\).
}

\begin{defn}[\(\cdes\)]
\label{def:cdes}
  The \emph{second-stage sample average treatment effect}, denoted \(\cdes\),
  is:
    \begin{equation}
    \label{eq:cdes}
      \cdes = \EE[Y(b, 1) - Y(w, 1) \mid M = 1].
    \end{equation}
\end{defn}

The estimand in Eq.~\eqref{eq:cdes} compares the potential second-stage
decisions under two race perception scenarios. For example, it compares the
potential charging decisions when the prosecutor perceives the individual to be
either Black or white; importantly, though, the estimand does not explicitly
consider the arresting officer's perception of race. Moreover, this estimand
restricts to the subset of individuals who had a ``positive'' first-stage
decision (e.g., those who were in reality arrested).

Because we condition on \(M = 1\) in the definition of the \(\cdes\), we may
equivalently write Eq.~\eqref{eq:cdes} as
\begin{equation}
  \label{eq:cdes_alt}
    \cdes = \EE[Y(b, M) - Y(w, M) \mid M = 1].
  \end{equation}
We can further write
  \begin{equation}
  \label{eq:cdes_second_alt}
    \cdes = \EE[Y(b) - Y(w) \mid M = 1],
  \end{equation}
where we define \(Y(z) = Y(z, M)\). Among those who reach the second stage
(i.e., individuals with \(M=1\)), \(Y(z) = Y(z, 1)\) denotes the outcome of
intervening \emph{only} on the second decider's perception of race. Among those
who do not reach the second stage (i.e., individuals with \(M=0\)), \(Y(z) =
Y(z, 0) = 0\).
Eqs.~\eqref{eq:cdes},~\eqref{eq:cdes_alt},~and~\eqref{eq:cdes_second_alt}, as
well as the informal estimand introduced at the beginning of
Section~\ref{sec:post_treat_bias}, are equivalent ways of capturing the same
underlying quantity, varying only in the degree to which they are explicit about
the staged nature of the process.

\subsection{Estimating the \texorpdfstring{\(\cdes\)}{CDE-Ob}}

Having defined the \(\cdes\), our goal is now to estimate it using only
second-stage data. That is, we aim to estimate the \(\cdes\) only using
observations for those individuals who received a ``positive''---and potentially
discriminatory---decision in the first stage. For example, we seek to estimate
discrimination in charging decisions based only on data describing those who
were arrested. As we show now, an ignorability assumption, together with an
overlap condition, is sufficient to guarantee the  \(\cdes\) is
nonparametrically identified by data on the second-stage decisions.

\begin{defn}[Subset ignorability]
\label{def:si}
  We say that \(Y(z, 1)\), \(Z\), \(M\), and \(X\) satisfy \emph{subset
  ignorability} if
    \begin{equation}
    \label{eq:ci}
      Y(z, 1) \indep Z \mid X, M=1
    \end{equation}
  for \(z\in\{w,b\}\).
\end{defn}
In our recurring example, subset ignorability means that among arrested
individuals, after conditioning on available covariates, race (as perceived by
the prosecutor) is independent of the potential outcomes for the charging
decision. As above, we can equivalently write Eq.~\eqref{eq:ci} as
  \begin{equation}
  \label{eq:ci_alt}
    Y(z) \indep Z \mid X, M=1.
  \end{equation}
This latter expression makes clear that subset ignorability is closely related
to the traditional ignorability assumption in causal inference, but where we
have explicitly referenced the first-stage outcomes to accommodate a staged
model of decision making.

In our prosecutorial setting, subset ignorability would fail if, for example,
there were a factor that prosecutors used to make their charging decisions but
which was not accounted for in the analysis (e.g., if prosecutors reviewed
witness statements that were not in the case files provided to the analyst),
and, further, that factor were unbalanced between groups (e.g., if all else
equal, witness statements were more commonly available in the cases of white
individuals). See Sections~\ref{sec:ex}~and~\ref{sec:empirical} for further
discussion of such unobserved confounders and their statistical consequences.

Almost all causal analyses implicitly rely on a version of subset ignorability,
since researchers rarely make inferences about the full population of interest.
For example, analyses are typically limited to the individuals who agreed to
participate in the study. Even randomized experiments, while ideal for internal
validity, frequently lack external validity because the study participants do
not resemble a larger population of interest. Whenever ascribing causal
interpretations to non-experimental data, it is important to carefully consider
the plausability of ignorability and other assumptions, as we discuss in detail
in Sections~\ref{sec:ex}~and~\ref{sec:empirical} below. We note, though, that
the assumptions we rely on are similar to those invoked in nearly every
observational study of causal effects.

Ignorability assumptions typically require a corresponding overlap condition to
guarantee consistent estimation.\footnote{%
  In the following, we assume that \(X\) is discrete for simplicity of
  exposition; for continuous analogues of these results, see
  Appendix~\ref{app:cont}.
}
\begin{defn}[Overlap]
  We say that \emph{overlap} holds when \(\Pr(Z = z \mid X = x, M = 1) > 0\) for
  all \(z\) and \(x\) such that \(\Pr(X = x, M = 1) > 0\).
\end{defn}

Overlap states that there are no covariate levels for which the probability of
receiving one of the treatments is zero within the population of interest. In
our prosecution example, overlap ensures that every case has a ``twin'',
identical in all aspects other than the stated race of the arrested individual,
against which it can be compared. Overlap would fail, therefore, in the
prosecutorial setting, if, for instance, there were alleged offenses for which
only Black individuals were arrested. We note that, unlike ignorability, overlap
can be assessed directly from the data; see Section~\ref{sec:empirical}. In
cases where overlap fails to hold, one can still elicit valid causal estimates
by restricting to the subset of the population where overlap holds. For example,
in assessing discrimination in prosecutorial charging decisions, one might only
consider those alleged offenses for which both Black and white individuals have
a non-zero probability of being arrested. But this restriction comes at the cost
of inferential validity for the original population. In such cases, one is
estimating the causal effect \emph{only} for the restricted population; the
causal effect for the original population may differ, sometimes substantially.

In the traditional, single-stage setting, ignorability and overlap are
sufficient to obtain consistent estimates of the average treatment effect.
Likewise, we now show that in our two-stage model of discrimination, subset
ignorability and overlap are sufficient to guarantee consistent estimates of the
\(\cdes\). In practice, if one can adjust for (nearly) all relevant factors
affecting second-stage decisions, one can (approximately) satisfy subset
ignorability, and in particular, one can estimate the \(\cdes\) only using data
available at the second stage. In the Appendix, we compare subset ignorability
to several alternatives, and show that those variants tend either to be too weak
to guarantee identifiability, or unnecessarily demanding for real-world
applications. We emphasize that since the first-stage decision, \(M\), and the
covariates, \(X\), can be viewed as pre-treatment relative to the second-stage
intervention, concerns about post-treatment bias corrupting estimates of the
\(\cdes\) are more naturally thought of as familiar concerns about
omitted-variable bias.\footnote{%
    See \cite{heckman1979sample} for related discussion on interpreting sample
    selection bias as omitted-variable bias.
}

\begin{thm}
\label{thm:main}
  Suppose \(Y(z,1)\), \(Z\), \(M\), and \(X\) satisfy subset ignorability and
  overlap. Then, the \(\cdes\) equals
    \begin{align*}
      \begin{split}
         & \sum_x \EE[Y \mid Z = b, X=x, M=1] \cdot \Pr(X=x \mid M = 1) \\
         & \hspace{1cm} - \sum_x \EE[Y \mid Z = w, X=x, M=1]
          \cdot \Pr(X=x \mid M = 1).
      \end{split}
    \end{align*}
\end{thm}

\begin{proof}
Conditioning on \(X\) in Eq.~\eqref{eq:cdes}, we have
    \begin{align}
      \begin{split}
      \label{eq:decomp}
        \cdes
          & = \sum_x \EE[Y(b, 1) \mid X = x, M = 1] \cdot \Pr(X = x\mid M = 1)\\
          & \hspace{1cm} - \sum_x \EE[Y(w, 1) \mid X = x, M = 1] \cdot
              \Pr(X = x\mid M = 1).
      \end{split}
    \end{align}
  By subset ignorability and overlap, we can condition the summands in
  Eq.~\eqref{eq:decomp} on \(Z = b\) and \(Z = w\), respectively, without
  changing their values, yielding
    \begin{align}
      \begin{split}
        \cdes
          & = \sum_x \B E[Y(b, 1) \mid Z = b, X = x, M = 1] \cdot
              \Pr(X = x\mid M = 1) \\
          & \hspace{1cm} - \sum_x \EE[Y(w, 1) \mid Z = w, X = x, M = 1] \cdot
              \Pr(X = x\mid M = 1)
      \end{split} \\
      \begin{split}
          & = \sum_x \B E[Y(Z, M) \mid Z = b, X = x, M = 1] \cdot
              \Pr(X = x\mid M = 1) \\
          & \hspace{1cm} - \sum_x \EE[Y(Z, M) \mid Z = w, X = x, M = 1] \cdot
              \Pr(X = x\mid M = 1).
      \end{split}
    \end{align}
  Finally, the statement of the proposition follows by consistency, as \(Y =
  Y(Z, M)\).
\end{proof}

\begin{cor}
\label{cor:main}
  Suppose subset ignorability and overlap hold, and that we have \(n\) i.i.d.\
  observations \((X_i, Z_i, Y_i)_{i=1}^n\) with \(M_i = 1\). Let \(S_x^{(n)} =
  \{1 \leq i \leq n : X_i = x\}\) represent the set of observations with \(X =
  x\), and let \(S_{zx}^{(n)} = \{1 \leq i \leq n : Z_i = z \land X_i = x\}\)
  represent the set of observations with \(X = x\) and \(Z = z\). Then the
  stratified difference-in-means estimator,
    \begin{equation}
    \label{eq:estimator}
      \Delta_n = \sum_x \left[
        \frac {1} {\big| S_{bx}^{(n)} \big|} \sum_{i \in S_{bx}^{(n)}} Y_i
      \right] \frac {\big| S_x^{(n)} \big|} {n} - \sum_x \left[
        \frac {1} {\big| S_{wx}^{(n)} \big|} \sum_{i \in S_{wx}^{(n)}} Y_i
      \right]\frac {\big| S_x^{(n)} \big|} {n},
    \end{equation}
  yields a consistent estimate of the \(\cdes\).
\end{cor}

\begin{proof}
Note that by the strong law of large numbers,
  \begin{align*}
    \lim_{n \rightarrow \infty} \frac {1} {\big| S_{zx}^{(n)} \big|} \sum_{i \in
    S_{zx}^{(n)}} Y_i
       &  \stackrel{\text{a.s.}}{=} \EE[Y \mid Z = z, X=x, M=1], \ \text{and} \\
    \lim_{n \rightarrow \infty} \frac{\big| S_x^{(n)} \big|}{n}
       & \stackrel{\text{a.s.}}{=} \Pr(X = x \mid M = 1).
  \end{align*}
Consequently,
  \begin{align*}
    \lim_{n \rightarrow \infty}\Delta_n &\stackrel{\text{a.s.}}{=}
      \sum_{x} \mathbb{E} [Y \mid Z = b, X = x, M = 1] \cdot \Pr(X = x \mid M =
      1) \\
      &\hspace{1cm} - \sum_x \mathbb{E} [Y \mid Z = w, X = x, M = 1] \cdot
        \Pr(X = x \mid M = 1),
  \end{align*}
which is the \(\cdes\), by Theorem~\ref{thm:main}.
\end{proof}

A straightforward calculation further shows that the following expression yields
a consistent estimate of the standard error of \(\Delta_n\):
  \begin{equation}
  \label{eq:sedim}
    \wh{\se}(\Delta_n) = \sqrt{\sum_x \left(
        \frac{\big| S_x^{(n)} \big|}{n}
      \right)^2 \left [
        \frac{c_{bx}(1-c_{bx})}{\big| S_{bx}^{(n)} \big|} +
        \frac{c_{wx}(1-c_{wx})}{\big| S_{wx}^{(n)} \big|}
      \right]},
  \end{equation}
where
  \begin{align*}
    c_{zx} = \frac {1} {\big| S_{zx}^{(n)} \big|} \sum_{i \in S_{zx}^{(n)}} Y_i.
  \end{align*}
Eq.~\eqref{eq:sedim} accordingly allows us to form confidence intervals for
\(\Delta_n\).

The nonparametric stratified difference-in-means estimator \(\Delta_n\) is the
basis for nearly all applications of benchmark analysis in discrimination
studies. In practice, as we discuss further in Section~\ref{sec:ex}, it is
common to approximate \(\Delta_n\) via a parametric regression model---but the
two estimators share the same theoretical underpinnings. As such, our analysis
above simply grounds traditional benchmark analysis within a specific causal
framework, and demonstrates that a particular ignorability assumption, together
with overlap, is sufficient to yield valid estimates.

\subsection{An alternative measure of discrimination}
\label{ssec:causal_comparison}

To better understand the \(\cdes\), we now contrast it with the total effect
(\(\te\))~\citep{imai-2010-general}, a second estimand considered by
discrimination researchers~\citep{knox-2019, heckman2020comment, zhao2020note}.
The total effect and the \(\cdes\) differ in our setting in two ways: (1) the
population of individuals about which we make inferences; and (2) the potential
outcomes being contrasted. The total effect is not restricted to individuals who
had a ``positive'' first-stage decision (e.g., it is not restricted to those who
were arrested). Additionally, we imagine a causal variable that reflects a
situation where the perception of race is counterfactually determined
\emph{before} the first-stage decision (instead of \emph{after} the first-stage
decision, as with the \(\cdes\)), and is the same at both stages.

We note that, in general---as discussed in Section~\ref{sec:motex} and
below---there is no fully coherent notion of a ``total effect'' of race, since
one cannot intervene on race, \emph{per se}. In our running example, the two
treatments (i.e., the officer's perception of race and the prosecutor's
perception of race) represent distinct, situation-specific notions of
intervening on race. In this restricted context, then, there is a natural
estimand that captures the spirit of a ``total effect'': comparing an
individual's potential outcomes had they been perceived as white or Black when
\emph{both} the first- and second-stage decisions were made. We formalize this
as follows:

\begin{defn}[\(\te\)]
\label{def:te}
  The \emph{total effect}, denoted \(\te\), is given by:
    \begin{equation}
    \label{eq:te}
      \te = \EE[Y(b, M(b)) - Y(w, M(w))].
    \end{equation}
\end{defn}

Unlike the \(\cdes\), which only measures discrimination in the second decision,
the total effect measures cumulative discrimination across \emph{both}
decisions. In our recurring example, the total effect captures the effect of
race at the time of arrest on the subsequent charging decision. In particular,
if a charged Black individual had instead been perceived as white by an officer,
they might never have been arrested, and hence never been at risk of being
charged, a possibility encompassed by the definition of the total effect, but
not by the \(\cdes\).

We stress, however, that in studies of discrimination---particularly racial
discrimination---there is often no clear intervention point, and the difference
between the \(\te\) and the \(\cdes\) is largely an artifact of how one defines
both the population of interest and the start of the decision-making process.
What is the \(\te\) in one description of events may be the \(\cdes\) in
another, equally valid description of the same events, as we describe next.

\begin{figure*}[t]
  \begin{center}
    \begin{tikzpicture}[
      minimum size=0.25cm,
      inner sep=0pt,
      yscale=0.35
    ]
      \draw[fill=blue1] (0, 10) node[anchor=north west, inner sep=5pt] {Spotted}
          rectangle (3, 0);
      \draw[fill=blue2] (3, 9) node[anchor=north west, inner sep=5pt] {Stopped}
          rectangle (6, 1);
      \draw[fill=blue3] (6, 8) node[anchor=north west, inner sep=5pt] {Arrested}
          rectangle (9, 2);
      \draw[fill=blue4] (9, 7) node[anchor=north west, inner sep=5pt] {Charged}
          rectangle (12, 3);

      \node[shape=circle, fill=white, draw] at (1.5, 2) (A1) {};
      \node[shape=circle, fill=white, draw] at (1.5, 3) (B1) {};
      \node[shape=circle, fill=white, draw] at (1.5, 4) (C1) {};
      \node[shape=circle, fill=black, draw] at (1.5, 5) (D1) {};
      \node[shape=circle, fill=black, draw] at (1.5, 6) (E1) {};
      \node[shape=circle, fill=black, draw] at (1.5, 7) (F1) {};

      \node[shape=circle, fill=white, draw] at (4.5, 3) (B2) {};
      \node[shape=circle, fill=white, draw] at (4.5, 4) (C2) {};
      \node[shape=circle, fill=black, draw] at (4.5, 5) (D2) {};
      \node[shape=circle, fill=black, draw] at (4.5, 6) (E2) {};

      \node[shape=circle, fill=white, draw] at (7.5, 4) (C3) {};
      \node[shape=circle, fill=black, draw] at (7.5, 5) (D3) {};
      \node[shape=circle, fill=black, draw] at (7.5, 6) (E3) {};

      \node[shape=circle, fill=black, draw] at (10.5, 5) (D4) {};

      \draw[->] (B1) -- (B2);
      \draw[->] (C1) -- (C2);
      \draw[->] (D1) -- (D2);
      \draw[->] (E1) -- (E2);

      \draw[->] (C2) -- (C3);
      \draw[->] (D2) -- (D3);
      \draw[->] (E2) -- (E3);

      \draw[->] (D3) -- (D4);

      \draw[thin, dotted] (0, 0) -- (0, -7);
      \draw[thin, dotted] (3, 0) -- (3, -7);
      \draw[thin, dotted] (6, 1) -- (6, -7);
      \draw[thin, dotted] (9, 2) -- (9, -7);
      \draw[thin, dotted] (12, 3) -- (12, -7);

      \draw[draw=blue2, fill=blue2] (3, -1.1) rectangle (6, -1.9);
      \draw[draw=blue3, fill=blue3] (6, -1.1) rectangle (9, -1.9);
      \draw[draw=blue4, fill=blue4] (9, -1.1) rectangle (12, -1.9);
      \draw[fill opacity=0] (3, -1.1) rectangle (12, -1.9);

      \draw[
        decorate, decoration={brace, amplitude=0.8em, aspect=7/12, mirror},
        yshift=-4pt
      ] (3, -2) -- (12, -2);
      \node at (7.5, -3.5) {\(\cdes = \te\)};
      \draw[
        decorate, decoration={brace, amplitude=1em, aspect=37/64}, yshift=4pt
      ] (0, -5) -- (12, -5);

      \draw[draw=blue1, fill=blue1] (0, -5.1) rectangle (3, -5.9);
      \draw[draw=blue2, pattern={crosshatch dots}, pattern color=white]
          (0, -5.1) rectangle (3, -5.9);
      \draw[draw=blue2, fill=blue2] (3, -5.1) rectangle (6, -5.9);
      \draw[draw=blue3, fill=blue3] (6, -5.1) rectangle (9, -5.9);
      \draw[draw=blue4, fill=blue4] (9, -5.1) rectangle (12, -5.9);
      \draw[fill opacity=0] (0, -5.1) rectangle (12, -5.9);

      \draw[thick, <->] (-1.5, -7) node[anchor=east] {\(\cdots\)} -- (13.5, -7)
          node[anchor=west] {\(\cdots\)};

      \draw (1.5, -7cm + 4pt) -- (1.5, -7cm - 4pt) node[anchor=north]
          {\(t_{k-1}\)};
      \draw (4.5, -7cm + 4pt) -- (4.5, -7cm - 4pt) node[anchor=north]
          {\(t_{k}\)};
      \draw (7.5, -7cm + 4pt) -- (7.5, -7cm - 4pt) node[anchor=north]
          {\(t_{k+1}\)};
      \draw (10.5, -7cm + 4pt) -- (10.5, -7cm - 4pt) node[anchor=north]
          {\(t_{k+2}\)};
    \end{tikzpicture}
  \end{center}
  \caption{\emph{%
    This figure illustrates estimands one could consider, and the populations
    they concern, as individuals move through one segment of the criminal
    justice system. For instance, one can measure combined discrimination in
    arrest and charging decisions either via the \(\te\) or the \(\cdes\). In
    studies of discrimination, there is no clear point at which race is
    ``assigned'' and so both the \(\te\) and the \(\cdes\) can be used
    interchangeably to express the same underlying causal effect, the \(\te\)
    with respect to the population of stopped individuals, and the \(\cdes\)
    with respect to the population of spotted individuals. More generally, the
    diagram illustrates a multistage process, where one seeks to measure
    discrimination culminating at stage \(t_{k + 2}\) (e.g., charging decisions)
    among those who make it to stage \(t_k\) (e.g., those who were stopped by
    the police). This quantity can be viewed as the \(\te\), where one imagines
    the process starting at time \(t_k\). Alternatively, it can be viewed as the
    \(\cdes\), where one views the process as starting earlier (at, say,
    \(t_{k-1}\), indicating that an officer spotted an individual), and then
    conditioning on those who made it to stage \(t_k\). Note that the
    quantities themselves are formally defined---and equivalent in the manner
    just described---even absent any considerations of estimation and
    randomization, which are not illustrated here.
  }}
\label{fig:cde-te}
\end{figure*}
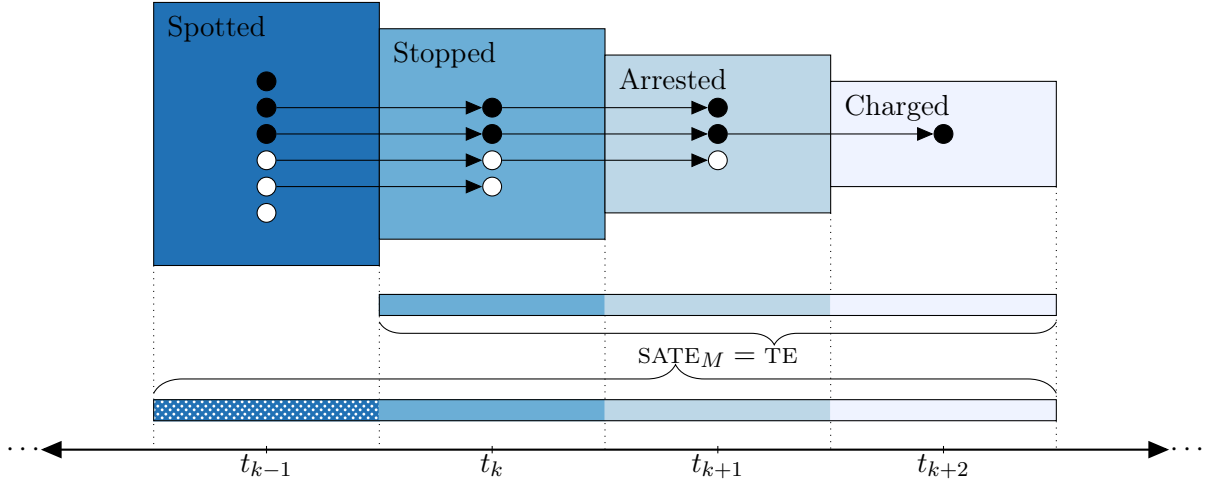

In our running example, the implicit population of interest consists of those
individuals stopped by the police, and the \(\te\) reflects a description of
events in which the decision-making process starts---and perception of race is
counterfactually determined---when the arrest decision is made. We can, however,
imagine moving back the clock and starting the process when the stop decision is
made, with the population of interest now comprising those individuals spotted
by an officer. In this case, the original \(\te\) is equivalent to the \(\cdes\)
on this newly defined population, where the first-stage decision indicates
whether an individual was \emph{stopped}. Both the original \(\te\) and the new
\(\cdes\) capture combined discrimination in the arrest and charging decisions,
among the subset of individuals who were stopped.\footnote{%
  To be explicit, our point is that the original \(\te\) and the new \(\cdes\)
  are the same quantities, and hence are estimable using the same data. However,
  the new \(\cdes\) (which subsets on individuals who are \emph{stopped} among
  those who are \emph{spotted}) and the original \(\cdes\) (which subsets on
  individuals who are \emph{arrested} among those who are \emph{stopped}), are,
  in contrast, not equal in general, and not necessarily estimable using the
  same data. In particular, if one wants to estimate either the original \(\te\)
  or, equivalently, the new \(\cdes\), the arrest decision can be viewed as an
  intermediate variable, and, accordingly, subsetting to arrested individuals
  would in general introduce post-treatment bias.
}

But the moment when an individual is spotted is no more statistically privileged
as a starting point than the moment when an officer makes a stop decision. One
could similarly measure cumulative discrimination that includes the stop
decision itself, either in terms of the \(\te\) or the \(\cdes\). For the
\(\te\),  as above, we imagine time starting immediately after a potential
police encounter, with the first-stage decision indicating whether an individual
was stopped (among a population of individuals spotted by the officer). For the
\(\cdes\), we back up the clock once again and imagine the first-stage decision
indicating whether an individual was spotted by an officer, among an even larger
population of people walking through the neighborhood where the officer patrols.
Figure~\ref{fig:cde-te} provides a graphical depiction of this
interchangeability.\footnote{%
  The formalism above shows a certain statistical equivalence between estimands
  having different starting points of the decision-making process. Nonetheless,
  the choice of starting point corresponds to measuring discrimination across
  different parts of the process, and so different estimands are relevant in
  different contexts. As such, we do not assert any normative ordering among
  them.
}

Although the \(\te\) may appear to avoid conditioning on intermediate outcomes,
it simply masks a complex chain of events that came before the nominal start of
the process, a chain that itself was likely influenced by discriminatory
decisions. For instance, the officer spotting and stopping motorists in our
running example could be patrolling the neighborhood in question because of its
racial composition.\footnote{%
  Importantly, even if the population of individuals spotted by police at a
  street corner is a (near) random sample of people living or working in the
  neighborhood, we still cannot think of race as being randomly assigned in that
  subset. In particular, spotted individuals may still differ on a variety of
  dimensions (e.g., socio-economic status) across race groups. As such, one
  would need to statistically account for these differences in any analysis that
  seeks to measure disparate treatment.
} The very idea of ``intermediate outcomes''---a concept
central to concerns about post-treatment bias---is a slippery notion in the
context of discrimination studies, where there is no clear point in time where
one can imagine that race is ``assigned.'' Even birth cannot be considered the
ultimate starting point since, in theory, one might include, at the least, the
race of a child's parents, determined at an earlier stage, when assessing
discrimination.\footnote{%
  In the case of biological sex, one might consider assignment to occur at
  conception, though that is typically not the primary moment of interest in
  studies of sex discrimination.
} Indeed, such generational counterfactuals may be critical for understanding
systemic, institutional discrimination.

Our discussion of discrimination in multi-staged, multi-decider scenarios
applies widely, but it is not universal. In particular, measuring discrimination
in a single-decider case---and, specifically, in officer use of force---is
challenging. In many of these single-decider scenarios, it is hard to imagine
intervening on race after the decision-making process begins, making it
difficult to isolate discrimination in later stages.

\section{Assessing Second-Stage Discrimination in a Stylized Scenario}
\label{sec:ex}

Subset ignorability, in theory, is sufficient to ensure nonparametrically
identified estimates of the \(\cdes\), even when the first-stage decisions are
discriminatory. We illustrate that idea by investigating in detail a
hypothetical scenario involving discriminatory arrest decisions in the first
stage and discriminatory charging decisions in the second stage. We explore the
properties of simple estimators in this setting through a simulation study. We
demonstrate that failing to adjust for a factor that directly influences
charging decisions can result in biased estimates of discrimination in those
decisions, but by accounting for all factors that directly influence charging
decisions---and hence satisfying subset ignorability---one can accurately
estimate the \(\cdes\), even when there is unmeasured confounding in arrest
decisions. This example further clarifies the conceptual importance of
distinguishing between an officer's perception of race and a prosecutor's
perception of race when defining and estimating our quantities of interest.

We consider a hypothetical jurisdiction in which police officers observe the
behavior and race of individuals who are potentially engaged in specific
criminal activity (e.g., a drug transaction) and then decide whether or not to
make an arrest. Subsequently, the case files of arrested
individuals---consisting of a written copy of the officer's description of the
encounter and the arrested individual's criminal history---are brought to a
prosecutor who decides whether or not to press charges. We assume the
prosecutor only observes the documented race and criminal record of the
arrestee, and the arresting officer's written description of the encounter;
accordingly, by construction, the charging decision depends only on these three
factors. For example, the prosecutor may choose only to charge individuals who
have several previous drug convictions and who were reported to be engaging in a
drug transaction. Importantly, while the prosecutor has access to an officer's
written report, the prosecutor does not directly observe the individual's
behavior leading up to the arrest.

Our goal is to estimate discrimination in charging decisions, formalized in
terms of the \(\cdes\). Intuitively, if we observe every arrested
individual's criminal history, race, and officer report, then subset
ignorability would hold because the prosecutor's charging decision depends only
on these factors. Thus, with these three covariates, we could generate valid
estimates of discrimination in prosecutorial decisions, even without knowing all
of the factors that led to an arrest, a decision that may itself have been
discriminatory. However, if any of these three covariates---criminal history,
race, or officer report---are unobserved, we will, in general, be unable to
accurately assess discrimination in prosecutorial decisions. In both scenarios,
with and without unmeasured confounding, our analysis is based on the
subpopulation of arrested individuals, where we note that the subsetting (i.e.,
arrest) is not influenced by the prosecutor's perception of race. In this
setting, the primary concern is thus omitted-variable bias, not post-treatment
bias.

We emphasize that we seek only to estimate discrimination in the second-stage
charging decision, not cumulative discrimination stemming from both the arrest
and charging decisions. In particular, while officer reports may represent an
inaccurate---and discriminatory---account of events, such discrimination is
distinct from that in the charging decision itself. Similarly, criminal
histories reflect a form of complex, long-term discrimination that we do not aim
to measure here. Alternative, and more expansive, notions of discrimination are
important to understand, but here we focus on assessing the prosecutor's narrow
contribution to inequities at a specific point in the process, a common
statistical objective closely tied to policy decisions and legal theories of
disparate treatment~\citep{jung2018omitted}.

\subsection{The data-generating process}
\label{sec:simulation}

\begin{figure}[t]
  \centering
    \begin{center}
    \begin{tikzpicture}[xscale = 3.5, yscale = 2, align = center]
      \node (race)      at (0, 0)
        {\(S\)\\{\footnotesize Self-Identified Race}};
      \node (behavior)  at (4/3, 1.8)
        {\(A\)\\{\footnotesize Behavior}};
      \node (race_o)    at (4/3, 0)
        {\(D\)\\{\footnotesize Officer-Perceived Race}};
      \node (arrest)    at (8/3, 0)
        {\colorbox{gray!75}{\(M\)}\\{\footnotesize Arrest}};
      \node (history)   at (2, -2)
        {\colorbox{gray!75}{\(X\)}\\{\footnotesize Criminal History}};
      \node (race_p)    at (8/3, -1)
        {\colorbox{gray!75}{\(Z\)}\\{\footnotesize Prosecutor-Perceived Race}};
      \node (report)    at (8/3, 1.8)
        {\colorbox{gray!25}{\(R\)}\\{\footnotesize Officer Report}};
      \node (charge)    at (4, 0)
        {\colorbox{gray!75}{\(Y\)}\\{\footnotesize Charge}};

      \draw[<->, dashed, bend left = 25]  (race)      to (behavior);
      \draw[->]                           (race)      to (race_o);
      \draw[<->, dashed, bend right = 30] (race)      to (history);
      \draw[->]                           (behavior)  to (arrest);
      \draw[->, bend left = 15]           (behavior)  to (report);
      \draw[->]                           (race_o)    to (arrest);
      \draw[line width = 5pt, white]      (race_o)    to (report);
      \draw[->]                           (race_o)    to (race_p);
      \draw[->]                           (race_o)    to (report);
      \draw[->]                           (race_o)    to (race_p);
      \draw[->, bend right = 30]          (history)   to (charge);
      \draw[->]                           (race_p)    to (charge);
      \draw[->]                           (arrest)    to (charge);
      \draw[->]                           (arrest)    to (report);
      \draw[->]                           (arrest)    to (race_p);
      \draw[->, bend left = 25]           (report)    to (charge);
    \end{tikzpicture}
  \end{center}
  \caption{\emph{%
      A causal DAG depicting our stylized example of arrest and charging
      decisions, where \(D\) represents the officer's perception of race, and
      \(Z\) represents the prosecutor's perception of race. Officer arrest
      decisions (\(M\)) are directly influenced by observed criminal behavior
      (\(A\)) and officer-perceived race (\(D\)); the officer reports of the
      encounters (\(R\)) are directly influenced by \(A\) and \(D\).
      Prosecutorial charging decisions are made for all arrested individuals,
      and are directly influenced by officer reports (\(R\)), criminal history
      (\(X\)), and prosecutor-perceived race (\(Z\)). Finally, an individual's
      self-identified race (\(S\)) influences the officer's perception of race
      (\(D\)), and is confounded with criminal history (\(X\)) and behavior
      (\(A\)). We consider two scenarios. The variables highlighted in dark gray
      (i.e., \(M\), \(Z\), \(X\), and \(Y\)) are always observed. In one
      scenario, the analyst also observes the officer report \(R\), highlighted
      in light gray, obtaining the full set of information available to the
      prosecutor; in the other, the analyst does not observe the officer report
      \(R\) (i.e., only \(M\), \(Z\), \(X\), and \(Y\) are observed), leading to
      omitted-variable bias.
  }}
\label{fig:sim-dag}
\end{figure}
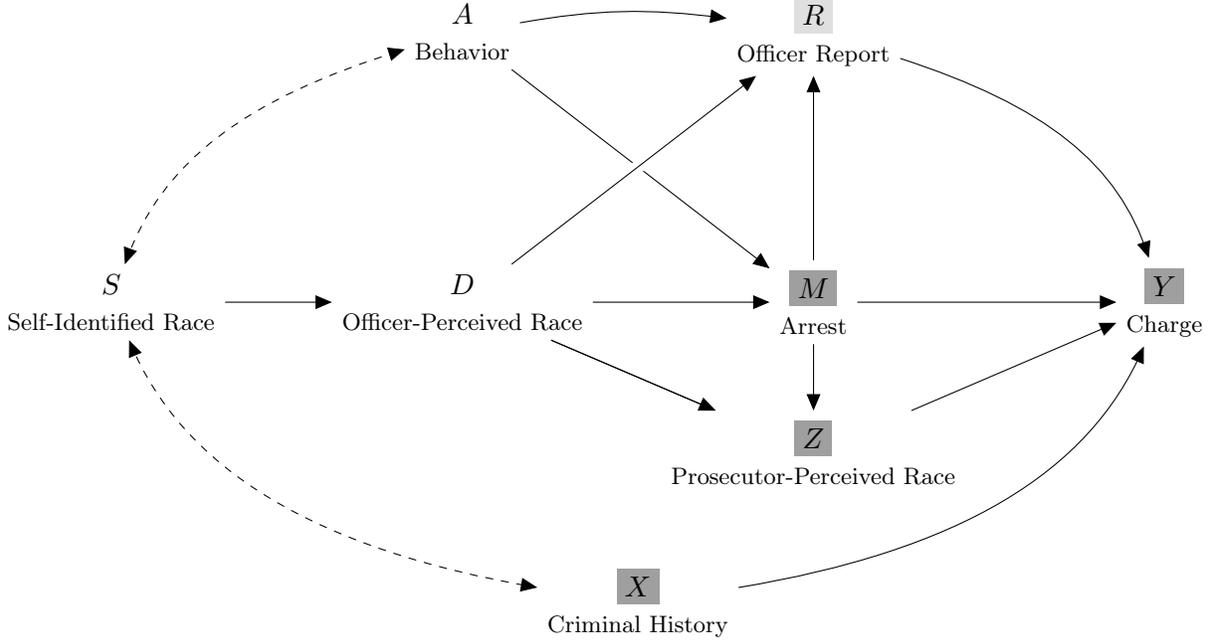

We now formally describe the data-generating process for our stylized example.
Under the structural causal model we consider, we can both compute the true
\(\cdes\) and compute estimates based only on select information available to
the prosecutor. In defining the generative process, we closely follow the
terminology and conventions of
\citet{pearl2009causality} and
\citet{pearl2016causal}.\footnote{%
  In particular, we follow \citet{pearl2009causality} in representing unobserved
  confounding by bidirectional dashed arrows; see Section~1.2.1. We do deviate
  from \citeauthor{pearl2009causality} in one aspect of our notation: we write
  counterfactuals as \(Y(z, m)\) instead of \(Y_{z, m}(u)\), suppressing the
  notational dependence on \(u\). The former notation aligns with the popular
  Rubin-Neyman potential outcome notation that we use when defining the
  \(\cdes\). We further note that this SEM is included primarily for
  illustrative purposes, and consequently contains some simplifications, such as
  strictly binary covariates. In practice, we recommend reasoning about subset
  ignorability and its relevant potential outcomes directly.
}

Our model is defined in terms of the causal directed acyclic graph (DAG)
depicted in Figure~\ref{fig:sim-dag}. In this model, \(S \in \{w,b\}\) indicates
one's self-identified race, and \(D\) and \(Z\) indicate, respectively, an
officer's and a prosecutor's perception of race. Further, \(M \in \{0,1\}\)
indicates the arrest decision, and \(Y \in \{0, 1\}\) indicates the charging
decision. Finally, \(A\) corresponds to an individual's behavior, as observed by
an officer, and \(X\) and \(R\) correspond, respectively, to criminal history
and an officer's description of an encounter, as included in the arrest report.
For simplicity, in our example these latter three variables are operationalized
as being binary---for example, one can imagine that \(X\) indicates whether an
individual had at least one previous drug conviction, \(A\) indicates whether
they were seen actively engaging in a drug transaction, and \(R\) indicates
whether they were reported by the officer to be actively engaging in a drug
transaction. Officers observe \(D\), \(A\), and \(R\) for all individuals;
prosecutors observe \(Z\), \(X\), and \(R\) only for the subset of arrested
individuals. Note that we also allow for \(Z\) and \(R\) to be missing (i.e., to
take the value \(\NA\)) in cases where an individual is not arrested.

Structural causal models are defined by a set of exogenous random variables and
deterministic structural equations specifying the values of all other variables
in the DAG. In our example, the independent exogenous variables are:
  \begin{align*}
    U_L & \sim \bern(\mu_{L}), \\
    U_A,\ U_M,\ U_X,\ U_R,\ U_Y & \sim \unif(0,1),
  \end{align*}
where \(\mu_L\) is an appropriately defined constant.

We define self-identified race (\(S\)), behavior (\(A\)), and criminal history
(\(X\)) in terms of \(U_L\), which captures latent confounding. For constants
\(\mu_A\), \(\gamma\), \(\mu_X\), and \(\delta\), the structural equations for
these three variables are given by:
  \begin{align*}
    f_S(u_L) & = \begin{cases}
        w & u_L = 0, \\
        b & u_L = 1,
    \end{cases} \\
        f_A(u_A, u_L) & = \B 1(u_A \leq \mu_A + \gamma \cdot u_L), \\
    f_X(u_X, u_L) & = \B 1(u_X \leq \mu_X + \delta \cdot u_L).
  \end{align*}
This specification allows for the distributions of criminal history and behavior
to vary by race due to exogenous factors like disparate police deployment and
historical discrimination. For example, stopped Black individuals may be less
likely to be engaged in criminal activity than stopped white individuals,
corresponding to \(\gamma < 0\).

In line with our discussion in Section~\ref{ssec:discrimination_model}, we set
the prosecutor's perception of race (\(Z\)) equal to the officer's perception of
race (\(D\)), and, for simplicity, we set both equal to one's self-identified
race (\(S\)). This choice yields the following structural equations:
  \begin{align*}
    f_D(s) & = s, \\
    f_Z(m, d) & = \begin{cases}
        d   & m = 1, \\
        \NA & m = 0. \\
    \end{cases}
  \end{align*}
Note that, when someone is not arrested, we represent the prosecutor's
perception of race as an explicit missing value. The arrest report, \(R\), is
treated similarly below.

Finally, for constants \(\alpha_0\), \(\alpha_A\), \(\alpha_{\text{black}}\),
\(\lambda_0\), \(\lambda_A\), \(\lambda_{\text{black}}\), \(\beta_0\),
\(\beta_X\), \(\beta_R\), and \(\beta_{\text{black}}\), the structural equations
for arrest decisions (\(M\)), police reports (\(R\)), and charging decisions
(\(Y\)) are given by:
  \begin{align*}
    &f_M(d, a, u_M) \\
      &\hspace{1cm}= \B 1(u_M \leq \alpha_0 + \alpha_A \cdot a +
      \alpha_{\text{black}} \cdot
      \B 1(d = b)), \\
    &f_R(d, a, m, u_R) \\
      &\hspace{1cm}= \begin{cases}
        \B 1(u_R \leq \lambda_0 + \lambda_A \cdot a + \lambda_{\text{black}}
          \cdot \B 1(d = b)) & m = 1, \\
        \mathrlap{\NA}{\hphantom{\B 1(u_Y \leq \beta_0 + \beta_X \cdot x +
          \beta_R \cdot r + \beta_{\text{black}} \cdot \B 1(z = b))}} & m = 0,
      \end{cases} \\
    &f_Y(z, m, r, x, u_Y) \\
      &\hspace{1cm}= \begin{cases}
        \B 1(u_Y \leq \beta_0 + \beta_X \cdot x + \beta_R \cdot r +
          \beta_{\text{black}} \cdot \B 1(z = b)) & m = 1 \land\, z \neq \NA
          \land\, r \neq \NA, \\
        \NA & m = 1 \land\, (z = \NA \lor\, r = \NA), \\
        0 & m = 0.
    \end{cases}
  \end{align*}
In particular, arrest decisions and police reports depend on an officer's
perception of race, whereas charging decisions depend on a prosecutor's
perception of race. This model incorporates both discrimination in arrest
decisions, via \(\alpha_{\text{black}}\), and discrimination in police
reports---e.g., by omitting potentially exculpatory details or by falsifying
information---via \(\lambda_{\text{black}}\). Discrimination in charging
decisions is encoded by \(\beta_{\text{black}}\).

The above structural equations, together with the distributions on the exogenous
variables, fully define the joint distribution of realized and potential
outcomes. In particular,
  \begin{align*}
    S &= f_S(U_L),          & D &= f_D(S),\\
    Z &= f_Z(M, D),         & A &= f_A(U_A, U_L),\\
    X &= f_X(U_X, U_L),     & M &= f_M(D, A, U_M),\\
    R &= f_R(D, A, M, U_R), & Y &= f_Y(Z, M, R, X, U_Y).
  \end{align*}

The primary causal quantity we seek to estimate---the \(\cdes\)---is defined in
terms of counterfactuals \(Y(z,m)\). As discussed in \cite{pearl2009causality}
and \cite{pearl2016causal}, such counterfactuals require some care to define, as
one must appropriately account for the exogenous variables \(U\). In particular,
for the causal DAG in Figure~\ref{fig:sim-dag}, the bivariate charge potential
outcomes, for counterfactual versions of prosecutor-perceived race, are given by
\(Y(z, m) = f_Y(z, m, R(m), X, U_Y)\), where \(R(m) = f_R(D, A, m, U_R)\) are
the counterfactual versions of the officer report. Further, the arrest potential
outcomes---where we consider counterfactual versions of officer-perceived
race---are given by \(M(d) = f_M(d, A, U_M)\). In general, counterfactuals
defined in this way obey the consistency rule, meaning that \(M = M(D)\) and \(Y
= Y(Z, M)\).

When \(\alpha_\text{black} \geq 0\), anyone who would be arrested if white would
also be arrested if Black (i.e., \(M(b) \geq M(w)\)). When
\(\alpha_\text{black} > 0\), we say arrest decisions are discriminatory since,
all else being equal, an individual is more likely to be arrested if they were
Black than if they were white. Likewise, \(Y(b, 1) \geq Y(w, 1)\) when
\(\beta_\text{black} \geq 0\), meaning that an individual who would be charged
if arrested and white would also be charged if arrested and Black. We say the
charging decision is discriminatory when \(\beta_\text{black} > 0\).

\begin{table}[t]
  \begin{center}
    \caption{\emph{%
      A sample of potential and realized outcomes for individuals in our
      hypothetical example. The data-generating process produces the full set of
      entries, but the prosecutor only observes the realized outcomes for those
      who were arrested, indicated by the shaded cells. In the first scenario we
      consider, the analyst also observes all the information in the shaded
      cells; in the second scenario, the analyst only observes the information
      in the dark gray cells (i.e., the analyst does not observe the officer
      report \(R\)), leading to omitted-variable bias.
    }}
    \label{tb:o}
    \begin{tabular}{@{\extracolsep{3pt}} cccccccccccccc}
      \\[-1.8ex]
      \hline \hline \\[-1.8ex]
      \(S\)
        & \(D\)
        & \(A\)
        & \(X\)
        & \(M(b)\)
        & \(M(w)\)
        & \(M\)
        & \(Z\)
        &\(R\)
        & \(R(0)\)
        & \(R(1)\)
        & \(Y(b, 1)\)
        & \(Y(w, 1)\)
        & \(Y\) \\
      \hline \\[-1.8ex]
      \(b\)
        & \(b\)
        & 0
        & 0
        & 0
        & 0
        & 0
        & \(\NA\)
        & \(\NA\)
        & \(\NA\)
        & 0
        & 0
        & 0
        & 0 \\
      \(b\)
        & \(b\)
        & 0
        & 1
        & 0
        & 0
        & 0
        & \(\NA\)
        & \(\NA\)
        & \(\NA\)
        & 1
        & 1
        & 0
        & 0 \\
      \(b\)
        & \(b\)
        & 1
        & \cellcolor{gray!75}1
        & 1
        & 0
        & \cellcolor{gray!75}1
        & \cellcolor{gray!75}\(b\)
        & \cellcolor{gray!25}0
        & \(\NA\)
        & 0
        & 1
        & 1
        & \cellcolor{gray!75}1 \\
      \(w\)
        & \(w\)
        & 0
        & \cellcolor{gray!75}0
        & 1
        & 1
        & \cellcolor{gray!75}1
        & \cellcolor{gray!75}\(w\)
        & \cellcolor{gray!25}1
        & \(\NA\)
        & 1
        & 0
        & 0
        & \cellcolor{gray!75}0 \\
      \(w\)
        & \(w\)
        & 0
        & 1
        & 0
        & 0
        & 0
        & \(\NA\)
        & \(\NA\)
        & \(\NA\)
        & 0
        & 0
        & 0
        & 0 \\
      \hline \\[-1.8ex]
    \end{tabular}

  \end{center}
\end{table}

\paragraph{Features of our data-generating process.}

Table~\ref{tb:o} displays a sample of five rows of data generated from our
model. From the full set of potential outcomes, we can compute the true
\(\cdes\) by directly applying Definition~1 to the generated data, taking the
average difference between \(Y(b, 1)\) and \(Y(w, 1)\) among arrested
individuals.\footnote{%
  Because \(Z\) and \(D\) are separately manipulable in our framing, this
  quantity---obtained by first subsetting on arrested individuals, and then
  computing the average difference between potential outcomes---can also be
  expressed in the \emph{do}-calculus: \(\cdes = \EE[Y \mid \Do(Z = b), M = 1] -
  \EE[Y \mid \Do(Z = w), M = 1]\). However, as is common in causal mediation
  analysis, if there were only one indecomposable treatment (e.g., if one
  instead imagined directly manipulating \(S\)) then the corresponding estimand
  could no longer be expressed using \emph{do}-operations alone
  \citep{pearl2009causality, pearl2015conditioning}. \label{fn:pearl}
}
However, given the simple linear form of our structural equations, a
straightforward calculation also shows that the \(\cdes\) is exactly equal to
\(\beta_{\text{black}}\).

Our hypothetical example captures three key features of real-world
discrimination studies. First, prosecutorial records do not contain all
information that influenced officers' first-stage arrest decisions (i.e.,
prosecutors only observe \(R\), not \(A\)). Second, our set-up allows for
situations where the arrest decisions are themselves discriminatory---those
where \(\alpha_\text{black} > 0\)---or the officer's report is discriminatory,
e.g., because of omission of exculpatory information or deliberate
falsification---those where \(\lambda_{\text{black}} > 0\). Third, the
prosecutor's records include the full set of information on which charging
decisions are based (i.e., \(Z\), \(X\), and \(R\)).

Among those who were arrested, the charging potential outcomes depend only on
one's criminal history (\(X\)) and the arrest report (\(R\)). In particular,
they do not depend on one's realized, prosecutor-perceived race (\(Z\)).
Consequently, \(Y(z, 1) \indep Z \mid X, R, M = 1\), meaning that the model
satisfies subset ignorability relative to \(X\) and \(R\). As a result, access
to \(X\) and \(R\), along with overlap, guarantees the stratified
difference-in-means is a consistent estimator of the \(\cdes\), even if one does
not have access to \(A\).\footnote{%
  In general, first-stage discrimination such as discriminatory arrest decisions
  or fabrication of evidence in arrest reports does not affect the consistency
  of the stratified difference-in-means estimator, since subset ignorability
  will continue to hold. Consistency may fail if discrimination is so extreme
  that overlap fails, e.g., if no white people are arrested.
} However, in general, \(Y(z, 1) \not \indep Z \mid X, M = 1\) (and, likewise,
\(Y(z, 1) \not \indep Z \mid R, M = 1\)), and so if one only has partial
information on charging decisions there is no guarantee the \(\cdes\) can be
consistently estimated.\footnote{%
  In the prosecutorial context, sufficiently diligent data gathering can
  mitigate this possibility; many offices maintain detailed case files, and we
  make use of such records in our empirical analysis in
  Section~\ref{sec:empirical}. In general studies of discrimination, it is
  important to ensure that decision factors are accurately captured and made
  available to analysts.
} Indeed, when there is such unmeasured
confounding in the prosecutor's decisions, one should expect biased estimates of
the \(\cdes\).

\subsection{Estimating the \texorpdfstring{\(\cdes\)}{CDE-Ob}}

Although the data-generating procedure produces the full set of potential
outcomes for each individual, the prosecutor only observes a subset of the
cells---realized outcomes for arrested individuals, highlighted in gray in
Table~\ref{tb:o}. While this circumscribes the causal effects one can
estimate---e.g., discrimination by police will no longer be identifiable in the
reduced dataset---one can still learn about the \(\cdes\). We explore the
performance of two statistical methods for estimating the \(\cdes\) based on
data observed by the prosecutor: the stratified difference-in-means estimator
described in Eq.~\eqref{eq:estimator}, and a regression-based estimator. We
apply each of these methods to two types of data: the full set of information
available to prosecutors (i.e., \(Y\), \(Z\), \(X\) and \(R\)), and an
incomplete dataset comprised only of \(Y\), \(Z\), and \(X\) (highlighted in
dark gray in Table~\ref{tb:o}), in which case we view \(R\) as an unmeasured
confounder.

One can compute the stratified difference-in-means estimate in three steps.
First,  partition arrested individuals into subsets that have the same value of
the available control variables (i.e., \(X\) and \(R\) in the complete data
setting, and \(X\) alone in the partial data setting). Second, on each resulting
subset, compute the average difference in charging rates between Black and white
individuals. Third, take a weighted average of these differences, where the
weights reflect the proportion of arrested individuals in each subset. In
addition, one can apply Eq.~\eqref{eq:sedim} to estimate the standard error of
this point estimate to generate confidence intervals.

The stratified difference-in-means estimator is theoretically appealing in that
it is guaranteed to yield consistent estimates of the \(\cdes\) when subset
ignorability and overlap hold. But the estimator can have high variance when the
dimension of the covariate space is high and the sample size is small. Thus, in
practice, it is common to model potential outcomes as a function of observed
covariates---also known as response surface modeling~\citep{hill2011bayesian}.
In particular, on the subset of arrested individuals, one can estimate the
\(\cdes\) via a parametric model that estimates observed charging decisions as a
function of the available information.

To demonstrate this latter approach, we use a linear probability model. In the
complete data setting, we have:
  \begin{align}
  \label{eq:linear-model}
    \EE[Y \mid Z, X, R] = \beta_0 +
    \beta_1 Z + \beta_2 X + \beta_3 R,
  \end{align}
where the model is fit on the full set of arrests seen by the prosecutor. Under
this model, the \(\cdes\) is approximated by the fitted coefficient
\(\hat{\beta}_1\), since that term captures the difference in charging potential
outcomes after adjusting for the observed covariates. For our specific stylized
example, the linear regression model in Eq.~\eqref{eq:linear-model} is in fact
perfectly specified---exactly mirroring the prosecutor's charging
decisions---and so we are guaranteed to obtain statistically consistent
estimates. In the partial data setting, where an analyst only has access to
\(X\), one must fit a reduced model that excludes \(R\):
  \begin{align}
  \label{eq:linear-model-incomplete}
    \EE[Y \mid Z, X] = \beta_0 +
    \beta_1 Z + \beta_2 X.
  \end{align}
In this case, \(\hat{\beta}_1\) in general yields a biased estimate of the
\(\cdes\), because of the omitted variable \(R\). The stratified
difference-in-means estimator will in general similarly yield a biased estimate
of the \(\cdes\) in this omitted-variable setting.

\subsection{Simulation results}
\label{ssec:sim_res}

We perform a simulation study to understand the properties of the above
estimators, varying our assumptions about discrimination and confounding. We
simulate 10,000 datasets of size 100,000 for each of 25 different parameter
settings. Each setting is defined as a combination of our two key
discrimination parameters, \(\alpha_\text{black}\) and \(\beta_\text{black}\),
where each parameter is allowed to take one of five values: 0.20, 0.25, 0.30,
0.35, and 0.40. Across all simulation settings, we assume the population of
individuals encountered by police is \(30\%\) Black (i.e., \(\mu_L = 0.3\));
that \(30\%\) of white individuals and \(40\%\) of Black individuals have a past
drug conviction, indicated by \(X\); and that \(30\%\) of white individuals and
\(20\%\) of Black individuals are seen engaging in a drug transaction, indicated
by \(A\).\footnote{%
  More specifically, the full set of parameters in our simulation was set as
  follows: \(\mu_L = 0.3, \mu_X = 0.3\), \(\mu_A = 0.3\), \(\delta = 0.1, \gamma
  = -0.1\), \(\alpha_0 = 0.1\), \(\alpha_A = 0.3\), \(\alpha_{\text{black}} \in
  \{0.2, 0.25, 0.3, 0.35, 0.4\}\), \(\lambda_0 = 0.2\), \(\lambda_A = 0.6\),
  \(\lambda_{\text{black}} = 0.1\), \(\beta_0 = 0.2\), \(\beta_X = 0.4\),
  \(\beta_R = 0.2\), and \(\beta_{\text{black}} \in \{0.2, 0.25, 0.3, 0.35,
  0.4\}\).
} These settings allow for a substantial amount of overlap across race groups
with regard to the key covariates.

\begin{figure}[t!]
  \begin{center}
    \includegraphics[height=3in]{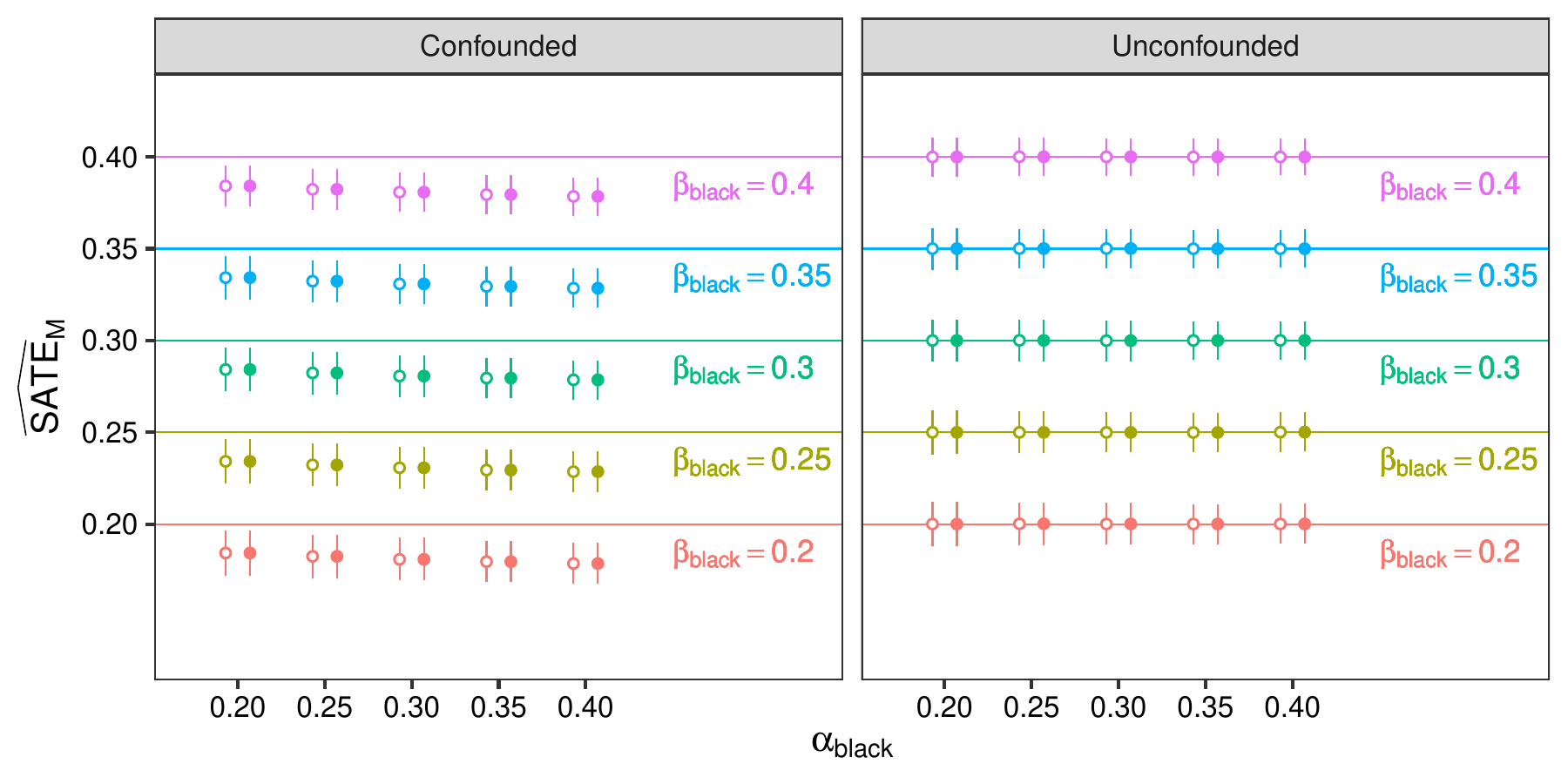}
    \caption{\emph{%
      In our hypothetical example of officer and prosecutor behavior, estimates
      of discrimination in charging decisions are biased when information
      directly influencing those decisions---in this case, an officer's
      report---is omitted (left). However, one can obtain accurate estimates of
      discrimination when accounting for all information directly influencing
      charging decisions (right). Each plot shows the results of 10,000
      simulations for each of 25 different combinations of discrimination in
      officer and prosecutor decisions, given by \(\alpha_\text{black}\) and
      \(\beta_\text{black}\), respectively. The true value of the \(\cdes\),
      indicated by the horizontal colored lines, is computed based on the full
      set of potential outcomes for each individual, and does not depend on the
      degree of discrimination in the first stage, as seen by the constant value
      of the \(\cdes\) across different values of \(\alpha_\text{black}\). For
      each parameter choice, we display the mean of the sampling distribution
      for the stratified difference-in-means estimator (solid circle) and the
      regression-based estimator (hollow circle), along with the interval
      spanned by the 2.5th and 97.5th percentiles of the sampling distribution.
      In the right plot (``unconfounded''), estimates are based on all three
      factors that directly influence charging decisions: race, criminal
      history, and officer report; in the left plot (``confounded''), we omit
      the report. When all variables directly influencing charging decisions are
      available, both estimators recover the true value of the \(\cdes\), even
      when there is an unknown degree of discrimination in arrest decisions.
    }}
    \label{fig:cdes_sampling}
  \end{center}
\end{figure}

On each synthetic dataset, we estimate the \(\cdes\) using both the stratified
difference-in-means estimator and the regression-based estimator, and compare
the results to the true population-level \(\cdes\) in two scenarios. To
illustrate the impact of omitted-variable bias, in the first scenario, we assume
the officer's report \(R\) is unavailable---meaning there is unmeasured
confounding---and therefore only stratify based on \(X\) in the
difference-in-means estimator, and fit the model in
Eq.~\eqref{eq:linear-model-incomplete} for the regression-based estimator. In
the second scenario, we assume that \(R\) is available, and stratify on both
\(X\) and \(R\) in the difference-in-means estimator, and fit the model in
Eq.~\eqref{eq:linear-model} for the regression-based estimator. For each
combination of \(\alpha_\text{black}\) and \(\beta_\text{black}\), the estimates
on the 10,000 synthetic datasets yield the approximate sampling distributions
for the difference-in-means and regression-based estimators. In
Figure~\ref{fig:cdes_sampling}, we summarize each sampling distribution by its
mean, 2.5th percentile, and 97.5th percentile. The solid points correspond to
the difference-in-means estimator, and the hollow points to the regression-based
estimator. The horizontal lines indicate the true population-level \(\cdes\).

In the left panel (``confounded'') of Figure~\ref{fig:cdes_sampling}, the points
lie below the horizontal lines in all cases, meaning we underestimate
discrimination in charging decisions. In this setting, estimates do not account
for the officer reports \(R\), and so there is unmeasured confounding in the
charging decisions. We set \(\gamma < 0\) in our simulations, and thus stopped
and arrested Black individuals are less likely to be engaging in criminal
activity, a pattern (noisily) reflected in the officer reports. Because we
assume these arrest reports are not available for analysis, we cannot fully
adjust for their direct influence on prosecutor decisions. As a result, by
adjusting for \(X\) alone, we miss an important, unmeasured difference between
arrested white and Black individuals, leading us to underestimate discrimination
in prosecutorial decisions.

In the right panel (``unconfounded'') of Figure~\ref{fig:cdes_sampling}, the
points lie on the horizontal lines in all cases, meaning the estimators are
unbiased, and the range between the 2.5th and 97.5th percentiles is relatively
narrow, indicating estimates are typically close to the true value. These
results hold even when one is unable to assess the degree of discrimination
\(\alpha_{\text{black}}\) in the arrest decisions. As implied by
Theorem~\ref{thm:main}, to accurately estimate the \(\cdes\), it is sufficient
to measure all covariates that directly influence the prosecutor's decisions. In
practice, it is nearly always impossible to do so perfectly; for instance,
decision factors such as forensic evidence may not be readily available, or
non-obvious factors, such as the time of day, may play a role in the
prosecutor's charging decision. Thus it is important to gauge the sensitivity of
estimates to unmeasured confounding in those decisions, as we demonstrate with
real-world data in Section~\ref{sec:empirical} below. The key point is that it
is sufficient to adjust for unmeasured confounding in the charging decisions
alone; to estimate discrimination in these charging decisions---formalized by
the \(\cdes\)---one need not account for unmeasured confounding in either the
documents generated by police, such as arrest reports, or the arrest decisions
themselves.

Finally, in addition to examining the sampling distributions, we assessed the
coverage of our \(95\%\) confidence intervals. For the difference-in-means
estimator, confidence intervals were constructed via the estimated standard
error given by Eq.~\eqref{eq:sedim}; and for the regression-based estimator, we
used the conventional OLS estimate of standard error. For each parameter
setting, we computed the proportion of confidence intervals for the 10,000
datasets that contained the true value of the \(\cdes\). In the no-confounding
scenario, we found the true coverage was in line with the nominal coverage,
ranging from \(94\%\) to \(96\%\) across parameter specifications. In the
confounding scenario, the intervals rarely covered the true values, as expected,
with coverage ranging from \(1\%\) to \(30\%\) across parameters.

\section{An Empirical Analysis of Prosecutorial Charging Decisions}
\label{sec:empirical}

We now apply the statistical framework developed above to assess possible race
and gender discrimination in real-world prosecutorial charging decisions. We
start with the set of individuals in a major U.S.\ county who were arrested for
a felony offense between 2013 and 2019. For our race-based analysis, we then
limit to the 25,918 instances in which the race of the arrested individual was
identified as either Black (14,686) or non-Hispanic white (11,232), and for our
gender-based analysis we limit to the 34,871 instances in which the gender of
the arrested individual was recorded as either male (29,283) or female
(5,588).\footnote{%
  Both Hispanic and non-Hispanic white individuals in our dataset appear to have
  been recorded simply as ``white''. To disentangle these two categories, we
  followed past work and imputed Hispanic ethnicity from
  surnames~\citep{word2008demographic, word1996building, OPP}.
}

Our dataset includes a variety of information about each case, including the
criminal history of the arrested individual; the alleged offenses (e.g.,
burglary); the location, date, and time of the incident; whether there is
body-worn camera footage; whether a weapon was involved; whether an elderly
victim was involved; and whether there was gang involvement. (See
Appendix~\ref{app:demographics} for additional details.) We also know the
ultimate charging decision for each case. Disaggregating by gender, \(51\%\) of
cases involving a male arrestee were charged, compared to \(45\%\) of cases
involving a female arrestee; and disaggregating by race, \(51\%\) of cases
involving a Black arrestee were charged, compared to \(50\%\) of cases involving
a white arrestee.

To gauge the extent to which charging decisions may suffer from disparate
treatment by race or gender, we estimate the \(\cdes\). We start by checking
that overlap is satisfied for both our race-based and our gender-based analyses.
Recall that overlap means \(\Pr(Z = z \mid X= x, M=1) > 0\), where \(Z=1\)
indicates an individual's ``treatment'' status (i.e., whether an individual is
male in our analysis of gender discrimination, or Black in our analysis of
racial discrimination), \(X\) is a vector of observed case features,  and \(M =
1\) means we restrict to those individuals who were arrested. In contrast to
ignorability, overlap can be assessed directly by examining the data. To do so,
we estimate propensity scores~\citep{rosenbaum1983central}, \(\Pr(Z = z \mid X=
x, M=1)\), via an \(L^1\)-regularized (lasso) logistic regression model. In
Figure~\ref{fig:overlap}, we plot the distribution of the estimated propensity
scores. In the left panel we disaggregate by gender, and in the right panel we
disaggregate by race (Black and white). In situations where overlap does not
hold, it is common to restrict one's analysis to a region of the covariate space
where it does hold. In our case, however, the vast majority of the data are
already far from the endpoints of the unit interval, so we work with the dataset
in its entirety.

\begin{figure}[t]
  \begin{center}
    \begin{subfigure}{.48\textwidth}
      \begin{center}
        \includegraphics[height=2in]{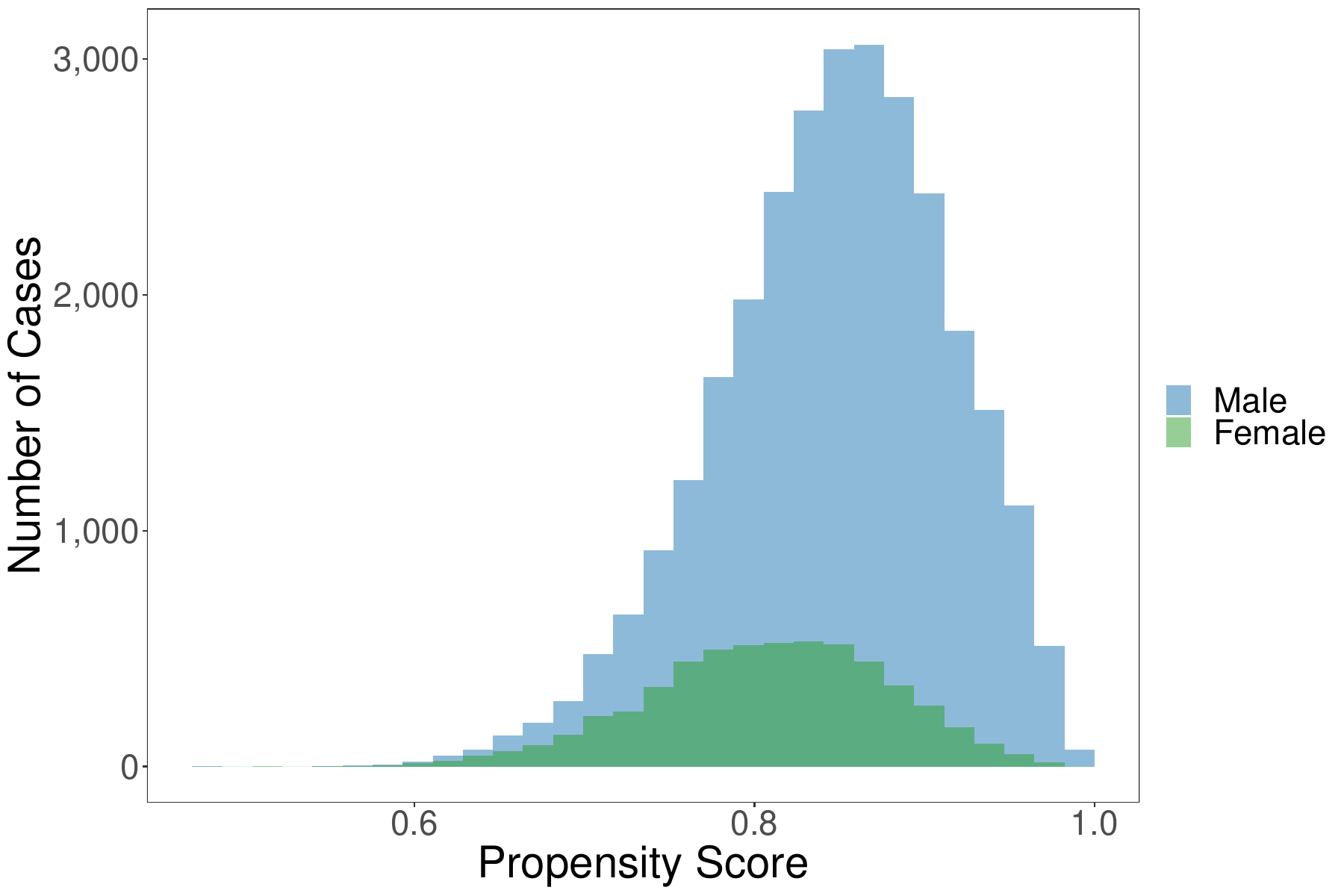}
        \caption{Gender-based analysis}
        \label{fig:gender_overlap}
      \end{center}
    \end{subfigure}
    \begin{subfigure}{.48\textwidth}
      \begin{center}
        \includegraphics[height=2in]{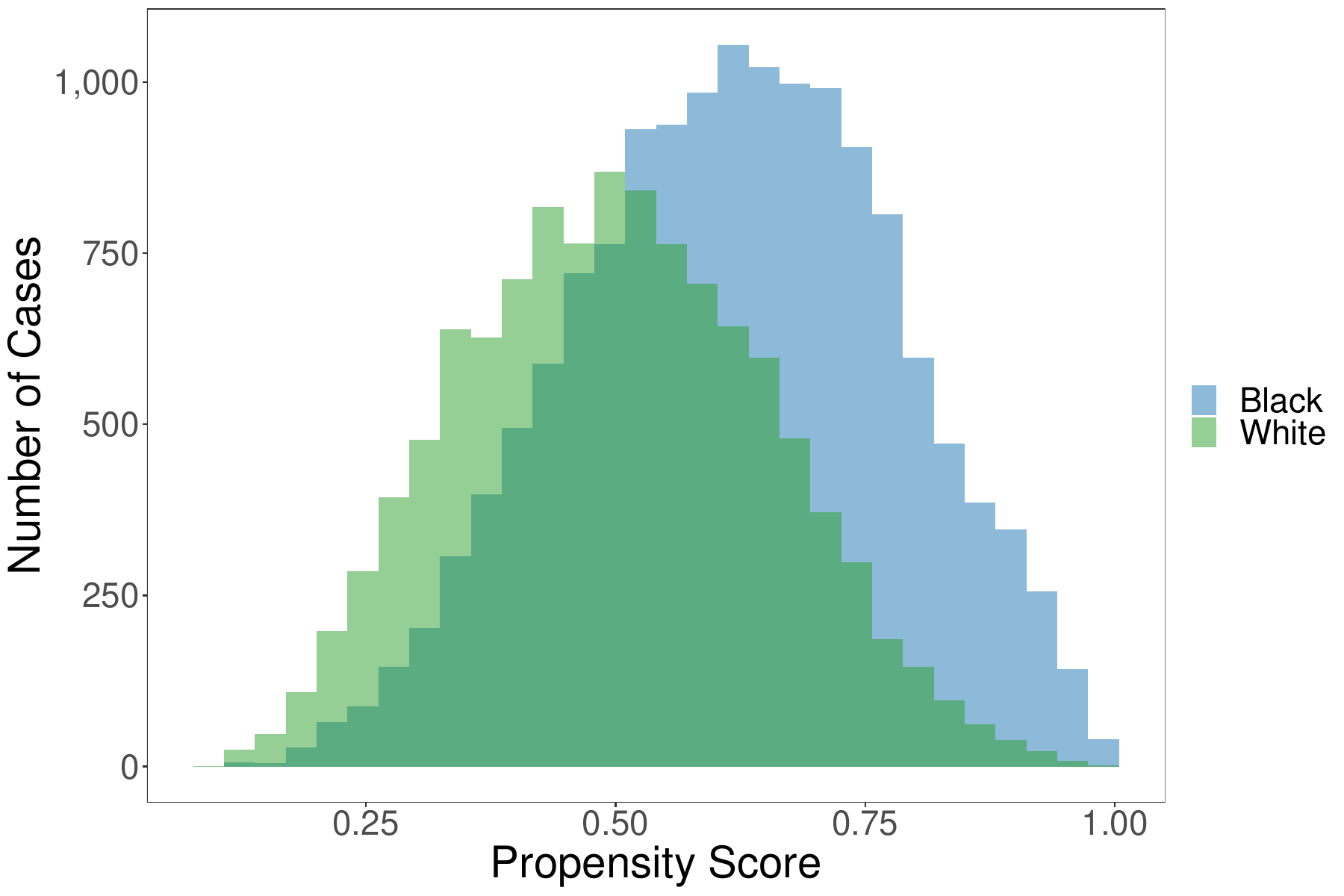}
        \caption{Race-based analysis}
        \label{fig:race_overlap}
      \end{center}
    \end{subfigure}
    \caption{\emph{%
      We plot, for both our gender-based (left) and race-based (right) analyses,
      the distribution of propensity scores, disaggregated by observed treatment
      status. We find that the propensity scores are concentrated away from the
      interval endpoints, satisfying overlap.
    }}
    \label{fig:overlap}
  \end{center}
\end{figure}

As discussed in Section~\ref{sec:ex}, regression-based estimators can be viewed
as a parametric variant of the stratified difference-in-means estimator
\(\Delta_n\). Thus, to help account for the high dimensionality of our feature
set, we now estimate the \(\cdes\) via linear regression. In particular, for
ease of interpretation, we use a linear probability model:
  \begin{equation}
  \label{eq:empirical_linear_model}
    \EE[Y \mid Z, X] = \beta_0 + \beta_1 Z + \beta_2^T X,
  \end{equation}
where \(Y\) indicates whether an arrested individual was charged, and \(X\)
denotes the vector of covariates.

In the gender model, we find that the \(\cdeshat\)---as given by
\(\hat{\beta}_1\)---is \(0.025\) (\(95\%\) CI: [\(0.014\), \(0.037\)]); and in
the race model, we find that the \(\cdeshat\) is \(-0.008\) (\(95\%\) CI:
[\(-0.018\), \(0.002\)]). These results indicate that the charging rate for men
is slightly higher than the rate for similar women, and that the charging rate
for Black individuals is on par with that of similar white individuals,
mirroring the patterns we saw with the raw, unadjusted charging rates. If there
are no unmeasured confounders (i.e., if subset ignorability holds) and our
parametric model is appropriate, these results suggest race and gender have a
relatively modest impact on charging decisions in the jurisdiction we consider.

To help contextualize these results, we note that past studies have found mixed
evidence of disparate treatment in prosecutorial charging decisions, likely due
in part to differences in the jurisdictions and time periods analyzed, and the
methods employed. In one of the most comprehensive investigations to date,
\citet{rehavi2014racial} examined nearly 40,000 individuals in the federal
criminal justice system from initial arrest to final sentencing. The authors
found that disparate treatment in prosecutorial charging
decisions---specifically for charges with statutory mandatory minimum
sentences---was a primary driver for sentencing disparities between Black and
white individuals. In contrast, in a recent experimental study,
\citet{robertson_race_2019} found no evidence of racial bias in charging
decisions when they presented prosecutors with vignettes in which the race of
the suspect was randomly varied. Similarly, in an observational analysis of
prosecutors at the San Francisco District Attorney's Office,
\citet{macdonald2017analysis} found little evidence of discrimination in
charging decisions---in fact, the authors found that white individuals were
charged slightly more often than similarly situated Black individuals. Finally,
in a recent quasi-random study of charging decisions at a large metropolitan
district attorney's office, \citet{chohlaswood2020blind} similarly found little
evidence of disparate treatment.

The AUC of our outcome model in Eq.~\eqref{eq:empirical_linear_model}
above---fit with all available covariates, including race and gender---is
\(86\%\), indicating that it can predict charging decisions well. Our model,
however, cannot capture all aspects of prosecutorial decision making, as at
least some information used by prosecutors (e.g., forensic evidence) is not
recorded in our dataset, meaning that subset ignorability likely does not hold
exactly. To check the robustness of our causal estimates to such unmeasured
confounding, one may use a variety of statistical methods for sensitivity
analysis~\citep{%
  rr, imbens2003sensitivity, carnegie2016assessing, dorie2016flexible,
  mccandless2007bayesian, mccandless2017comparison, jung2020bayesian,
  franks2019flexible%
}. At a high level, these methods posit relationships between
the unmeasured confounder and both the treatment variable (e.g., race or gender)
and the outcome (e.g., the charging decision), and then examine the sensitivity
of estimates under the model of confounding.

We apply a technique for sensitivity analysis recently introduced by
\citet{cinelli2018making}. In brief, their approach bounds the extent to which
a coefficient estimate in a linear model---like \(\hat{\beta}_1\) in
Eq.~\eqref{eq:empirical_linear_model}---might change if one were to refit the
model including an unmeasured confounder \(U\). More specifically, under the
extended model
  \begin{equation*}
    \EE[Y \mid Z, X, U] = \beta_0 + \beta_1 Z +  \beta_2^T X + \gamma U,
  \end{equation*}
\citeauthor{cinelli2018making} bound the change in \(\hat{\beta}_1\) in terms of
two partial \(R^2\) values: \(\rsqy\) and \(\rsqz\). These two values
respectively quantify how much residual variance in the outcome \(Y\) and
treatment \(Z\) is explained by \(U\). Formally, \(\rsqy\) is defined in terms
of the \(R^2\) of two linear regressions: one using all the covariates \(X\),
\(Z\), and \(U\) to estimate \(Y\) (\(R^2_{\text{full}}\)), and one excluding
\(U\) (\(R^2_{\text{red}}\)). Then, \(\rsqy = (R^2_{\text{full}} -
R^2_{\text{red}}) / (1 - R^2_{\text{red}})\). The quantity \(\rsqz\) is defined
analogously. As these partial \(R^2\) values increase, so does the amount by
which \(\hat{\beta}_2\) could change.

\begin{figure*}[t]
  \begin{center}
    \begin{subfigure}{.48\textwidth}
      \begin{center}
        \includegraphics[height=2.75in]{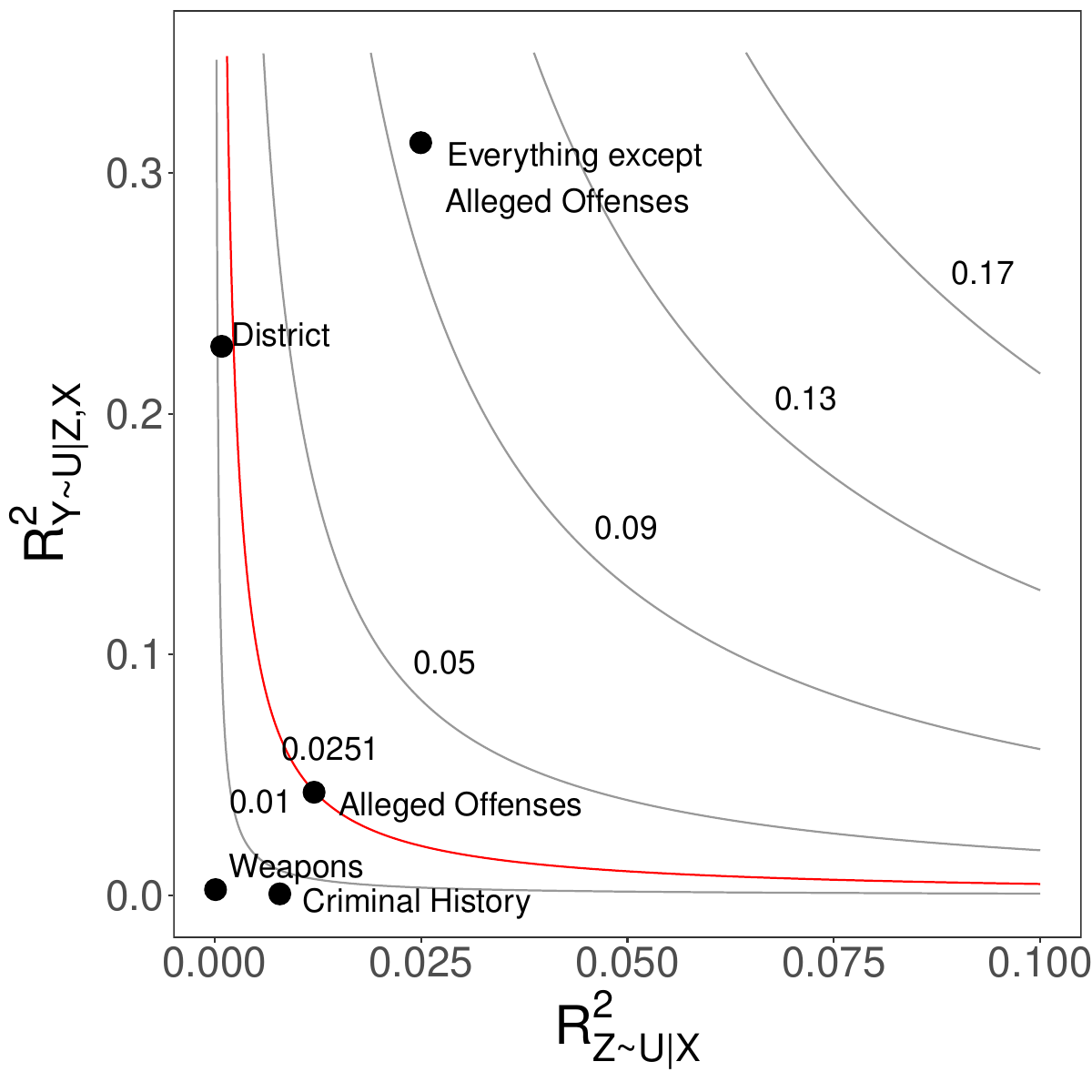}
        \caption{Gender-based analysis}
      \end{center}
    \label{fig:sens_gender}
    \end{subfigure}
    \begin{subfigure}{.48\textwidth}
      \begin{center}
        \includegraphics[height=2.75in]{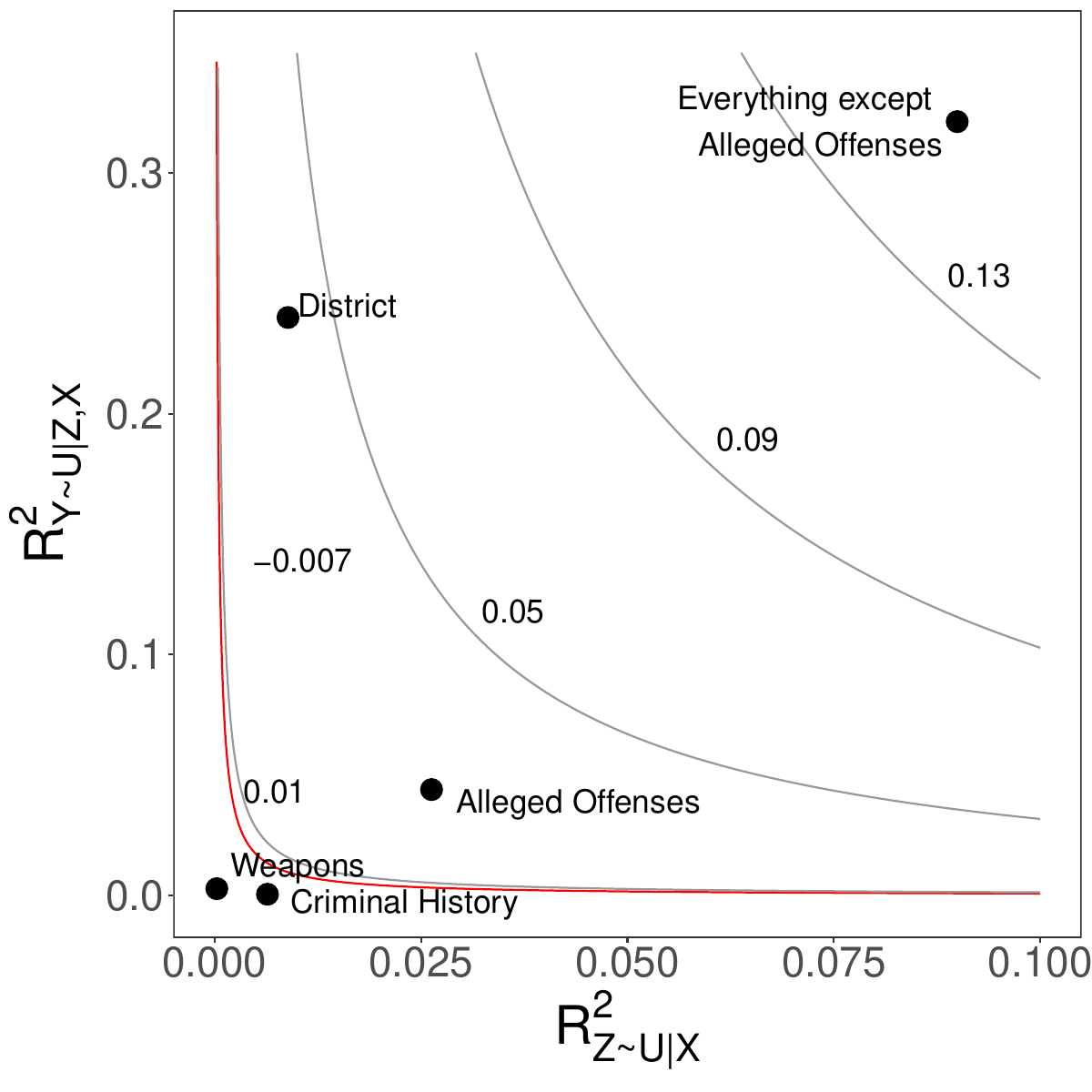}
        \caption{Race-based analysis}
      \end{center}
    \label{fig:sens_race}
    \end{subfigure}
    \caption{\emph{%
      Contour plots describing the sensitivity of the \(\cdeshat\) to unmeasured
      confounding, for our analysis of gender (left) and race (right). The plots
      indicate the maximum amount the \(\cdeshat\) may change under the
      \citet{cinelli2018making} model of confounding, parameterized by two
      partial \(R^2\) values. The red curves correspond to a change equalling
      the magnitude of the \(\cdeshat\) estimated from the available data. Thus,
      an unobserved confounder corresponding to a point above the red curve
      would be capable of changing the sign of our estimate. To aid
      interpretation, both plots display the partial \(R^2\) values associated
      with several observed subsets of covariates.
    }}
    \label{fig:sens}
  \end{center}
\end{figure*}

The contour plots in Figure~\ref{fig:sens} show the maximum amount by which the
\(\cdeshat\) may change as a function of \(\rsqy\) and \(\rsqz\) for our
analysis of gender and race---with that change potentially increasing or
decreasing the estimate. The red lines trace out values for which the maximum
change equals our empirical point estimates of the \(\cdeshat\). In particular,
an unmeasured confounder lying above the red line could be sufficient to change
the sign of our estimate.

A key hurdle in sensitivity analysis is positing a reasonable range for the
strength of a possible unmeasured confounder. To aid interpretation, we compute
the partial \(R^2\) values for various subsets of observed covariates, as
recommended by \citeauthor{cinelli2018making}. For each such subset, we fit the
regression model in Eq.~\eqref{eq:empirical_linear_model} both with and without
that subset, which in turn yields a pair of partial \(R^2\) values for that
subset of covariates.

The contour plots in Figure~\ref{fig:sens} contain these reference points for
five different subsets of covariates: (1) the subset describing criminal history
(e.g., number of prior convictions and number of prior arrests); (2) the alleged
offenses (e.g., burglary); (3) the subset of all covariates except for the
alleged offenses; (4) the district in which the alleged incident took place; and
(5) whether a weapon was alleged to have been used. We find that the partial
\(R^2\) values associated with criminal history and whether a weapon was used
are below the red curves for both our analysis of gender and race,
indicating that a confounder with comparable marginal explanatory power to these
covariates would not be sufficient to change the sign of our estimates.
However, the partial \(R^2\) values corresponding to the alleged offenses and
the district in which the charges were filed are near the red curve for our
gender-based analysis and far above the curve for our race-based analysis,
meaning that omitting a covariate with similar explanatory power could
qualitatively change our conclusions. Furthermore, the partial \(R^2\) values
corresponding to everything except the alleged offenses are far above the red
curve in both cases, suggesting that an unobserved confounder of similar
strength could again substantially alter our results. For instance, in this
extreme scenario, inclusion of a currently omitted confounder with similar
characteristics in the race-based analysis could yield an estimated treatment
effect of more than \(13\%\).

One cannot know the exact nature and impact of unmeasured confounding. Thus, as
in many applied statistical problems, we must rely in large part on domain
expertise and intuition to form reasonable conclusions. In this case, given the
results of our sensitivity analysis, we interpret our empirical findings as
providing evidence that perceived gender and race have limited effects on
prosecutorial charging decisions in the jurisdiction we consider. As with the
\(\cdes\), our sensitivity analysis is solely focused on discrimination in the
charging decision, and, in particular, is not designed to capture the cumulative
effects of discrimination stemming from arrests and other earlier decision
points.

\section{Discussion}
\label{sec:discussion}

We have outlined a formal causal framework to ground observational studies of
discrimination. We specifically showed that subset ignorability, together with
overlap, is sufficient to guarantee that one important causal measure of
discrimination (the \(\cdes\)) is nonparametrically identified in a canonical
two-stage decision-making setting. \emph{In this context}, we therefore believe
potential issues of post-treatment bias are more appropriately thought of as
concerns about omitted variables. Indeed, our treatment of interest---perception
of race by the second decision maker---occurs after the subsetting in the first
stage, and so it is not post-treatment relative to the selection process. As
such, we demonstrated that a traditional regression-based analysis can be used
to assess discrimination in real-world prosecutorial charging decisions, even
though the underlying arrests may have been discriminatory in unknown ways. In
that example---as in many applied settings---subset ignorability may only hold
approximately, and our empirical analysis illustrates the importance of
sensitivity analysis for robust inference.

Measurements of the \(\cdes\) can be an important step in quantifying
discrimination by specific decision makers at specific points in time. In our
running example, estimates of the \(\cdes\) can help identify prosecutors who
may be making systematically biased charging decisions. Identification of bias,
however, is only the first step toward reform. To mitigate identified
disparities, one could imagine a variety of interventions, such as training
programs \citep{spencer2016implicit}, or blinding prosecutors to the race of
arrested individuals \citep{chohlaswood2020blind}. As with all interventions,
care must be taken to ensure they do not have unintended consequences. Changes
in prosecutorial policies could have negative spillover, for example on
policing, or unexpected equilibrium effects, such as overall harsher charging
decisions.

The \(\cdes\) is but one way to characterize and inform interventions designed
to reduce discriminatory behavior. There are at least two broad notions of
discrimination, which approximately map to the legal concepts of disparate
treatment and disparate impact. Both involve causal interpretations, though with
key differences in the definition of the estimand. Disparate treatment concerns
the causal effect of race on outcomes---as we formalize here by the
\(\cdes\)---with behavior often driven by animus or explicit racial
categorization. Disparate impact, on the other hand, concerns the causal effect
of policies or practices on unjustified racial disparities, regardless of
intent. Disparate treatment and disparate impact both play important roles in
legal and policy discussions, and the perspective one adopts in any given
situation affects the choice of statistical estimation strategy and the
interpretation of results~\citep{jung2018omitted}.

We have throughout focused on the statistical foundations and measurement of
disparate treatment. In our primary example, we estimate---assuming subset
ignorability holds---that perceived race and gender have relatively small
effects on prosecutorial charging decisions in the jurisdiction we examine. We
further demonstrate that these estimates are moderately robust to potential
omitted-variable bias. However, that finding, in and of itself, does not mean
charging decisions are equitable in a broader sense. Consider, for example, the
1,637 cases in our data involving alleged possession of controlled substances by
Black or non-Hispanic white individuals. Of these, 748 cases (\(46\%\)) were
ultimately charged, and charging rates by race were nearly identical across race
groups, offering little prima facie evidence of disparate treatment. However,
among the 748 charged cases, 464 (\(62\%\)) involved a Black individual---far
exceeding the proportion of Black residents in the county we study. Charging
decisions for these cases thus impose a heavy burden on Black individuals, even
if those decisions were not tainted by animus. To the extent that prosecution of
drug crimes is misaligned with community goals, these decisions create an
unjustified, and discriminatory, disparate impact.

Rigorously estimating discrimination is a daunting task that requires careful
consideration. At an empirical level, it is often difficult to obtain detailed
data on individual decisions, in which case benchmark analysis may be
inadequate---even if coupled with sensitivity analysis. At a theoretical level,
we have a limited statistical language to make precise concepts such as animus
and implicit bias that are central to discrimination research. Further, as we
note above, past work has often framed discrimination as the causal effect of
race on behavior, but other conceptions of discrimination, such as disparate
impact, are equally important for assessing and reforming practices. Finally,
the conclusions of discrimination studies are generally limited to specific
decisions that happen within a long chain of potentially discriminatory actions.
Quantifying discrimination at any one point (e.g., in charging decisions) does
not yield estimates of specific or cumulative discrimination at other points
(e.g., in arrest decisions). Despite these important considerations, we hope our
work helps place discrimination research on more solid statistical footing, and
provokes further interest in the subtle conceptual and methodological issues at
the heart of discrimination studies.

\bibliographystyle{abbrvnat}
\bibliography{refs}

\newpage
\appendix

\setcounter{table}{0}
\setcounter{figure}{0}
\setcounter{equation}{0}

\renewcommand\thetable{\thesection.\arabic{figure}}
\renewcommand\thefigure{\thesection.\arabic{figure}}
\renewcommand{\theequation}{\thesection.\arabic{equation}}

\section{A Comparison to Alternative Ignorability Conditions}
\label{sec:ignorability}

To better understand subset ignorability, we compare it to alternative
conditions that recently have been proposed in the context of discrimination
studies. In particular, we compare subset ignorability to a set of assumptions
introduced by \citet{knox-2019}, which they call treatment ignorability,
mediator ignorability, and mediator monotonicity. We show that this set of
assumptions, like subset ignorability, is sufficient---but not necessary---to
ensure the \(\cdes\) is nonparametrically identified by data on second-stage
decisions. Importantly, however, the \citeauthor{knox-2019}\ conditions are
unlikely to be satisfied in important examples of potentially discriminatory
decision making where subset ignorability holds (either exactly or
approximately) and the \(\cdes\) accordingly can be estimated, like those
situations presented in Sections~\ref{sec:ex} and \ref{sec:empirical}.

Aside from the \citeauthor{knox-2019}\ conditions, it is instructive to compare
subset ignorability to sequential ignorability~\citep{imai-2010-general,
imai-2010-identification}, a popular and often useful concept that was
introduced to formalize causal mediation analysis, and one that is closely
related to the \citeauthor{knox-2019}\ conditions. Sequential ignorability is
strictly stronger than subset ignorability, meaning that the former implies the
latter but that the converse does not hold. In the setting of discrimination
studies, there is little reason to believe sequential ignorability---or
reasonable approximations of it---would be satisfied, and we primarily discuss
the idea to clarify its distinction from subset ignorability.

The alternative ignorability conditions considered here were developed in the
context of a single treatment. Therefore, to facilitate a direct comparison
between subset ignorability and these alternatives, we adopt this
single-treatment perspective throughout the Appendix. As discussed in the main
text, there are substantive issues with positing a single manipulation of
(perceived) race, gender, or other immutable characteristics in many multi-stage
settings. Formally, however, it is straightforward to collapse \(Z\) and \(D\)
to a single treatment---which we call \(Z\)---that affects both the first-stage
and the second-stage decisions. In particular, we now assume the potential
outcomes \(M(z)\) and \(Y(z, m)\) satisfy the consistency relations \(M = M(Z)\)
and \(Y = Y(Z, M)\). We emphasize that in this new framing, the definition of
subset ignorability in Eq.~\eqref{eq:ci} remains the same and that
Theorem~\ref{thm:main} likewise holds unaltered---since neither explicitly
references the first-stage potential outcomes.\footnote{%
  As noted in Footnote~\ref{fn:pearl} above, since we have restricted to the
  context of a single treatment \(Z\), many of the quantities we consider are
  not expressible via the \emph{do}-calculus, though they are still expressible
  in terms of potential outcomes. We emphasize that these potential outcomes
  should be understood in the conventional sense~\citep{pearl2009causality}:
  \(Y(z, 1)\) represents what would have resulted for an individual if,
  counterfactually, one had intervened on \(M\) so that \(M = 1\) and \(Z\) so
  that \(Z = z\). Although directly manipulating the first-stage decision so
  that \(M = 1\) may be implausible in some situations---for instance, it may be
  challenging in practice to intervene on an arresting officer's decision---no
  issue arises in our setting as we are only concerned with the outcomes \(Y(z,
  1)\) for individuals who would be arrested in the absence of such an
  intervention (i.e., where it is already the case that \(M = 1\)). Moreover,
  while the FFRCISTG framework \citep{robins1986ffrcistg, richardson2013swig}
  may consider these to be  ``cross-world'' counterfactual quantities, we note
  that recent extensions of these frameworks discussed in
  \citet{robins2020interventionist} could accommodate our estimand and
  identifying assumptions by allowing for the race variable to be split into
  race variables that are time- and context-specific, as we did in the main body
  of the paper.
}

We start by formally considering sequential ignorability, following
\cite{imai-2010-general, imai-2010-identification}.

\begin{defn}[Sequential ignorability]
  We say that \emph{sequential ignorability} is satisfied when the following two
  conditional independence criteria hold:
    \begin{align}
      \label{eq:si1} \{Y(z', m), M(z)\} \indep Z & \mid X,\\
      \label{eq:si2} Y(z', m) \indep M & \mid Z, X,
    \end{align}
  for \(z, z' \in \{w, b\}\) and \(m\in\{0, 1\}\).
\end{defn}

The two key conditional independence assumptions we list are the same as in the
definition of sequential ignorability given by \cite{imai-2010-general,
imai-2010-identification}, but to facilitate direct comparison with other
ignorability criteria, we omit from our definition the accompanying overlap
conditions. Also, for ease of exposition, we present the definition in the
setting of binary treatment and mediator variables, though the original was more
general. In the context of our running example, sequential ignorability means
that: (1) conditional on the observed covariates \(X\), the potential outcomes
for charging \(Y(z,m)\) and arrest \(M(z)\) are jointly independent of an
individual's actual race \(Z\); and (2) conditional on the observed covariates
\(X\) and an individual's race \(Z\), the arrest decision \(M\) is independent
of the potential charging outcomes \(Y(z,m)\).

Theorem~\ref{thm:main-body}, below, shows that sequential ignorability implies
subset ignorability, but also, importantly, that sequential ignorability is a
strictly stronger condition. To understand why, consider the stylized model of
Section~\ref{sec:simulation}, in which one has all of the information that
drives a prosecutor's charging decision---satisfying subset ignorability---but
not all of the information that drives an officer's arrest decision. For
example, suppose the prosecutor has access to the officer's report, but not the
arrested individual's actual behavior. In this case, one would in general
expect the first condition of sequential ignorability---in
Eq.~\eqref{eq:si1}---to be violated. In particular, without detailed data on
what an officer observes, there is little reason to think the arrest potential
outcomes, \(M(z)\), would be independent of an individual's race, even
controlling for factors available to the prosecutor.

We next formally present the definitions of treatment ignorability, mediator
ignorability, and mediator monotonicity proposed by \citeauthor{knox-2019},
starting with treatment ignorability.
\begin{defn}[Treatment ignorability]
  \emph{Treatment ignorability} is the combination of the following two
  conditional independence criteria: for \(z, z' \in \{w,b\}\) and \(m \in
  \{0,1\}\),
    \begin{align}
      \label{eq:ti1} M(z) \indep Z & \mid X, \\
      \label{eq:ti2} Y(z', m) \indep Z & \mid M(w), M(b), X.
    \end{align}
\end{defn}

In the context of arrest and charging decisions, treatment ignorability means
that: (1) the potential outcomes for the arrest decision \(M(z)\) are
independent of race \(Z\), after conditioning on the observed covariates \(X\);
and (2) the potential outcomes for the charging decision \(Y(z', m)\) are
independent of race \(Z\) after conditioning on both the covariates \(X\) and
the arrest potential outcomes \(M(w)\) and \(M(b)\).

The first condition of treatment ignorability is similar to the first condition
of sequential ignorability, and it is unlikely to hold in our setting for the
same reason. In general, given only information about what motivates the
second-stage decision (e.g., charging, in our case) one cannot say much about
what occurs in the first stage (e.g., arrest). But, critically, such
information about the first stage is not necessary to estimate the \(\cdes\),
which only quantifies discrimination in the second-stage decision.
Theorem~\ref{thm:main} makes that statement precise, showing that subset
ignorability---which does not consider first-stage potential outcomes---is
sufficient to ensure the \(\cdes\) is nonparametrically identified by the
second-stage data.

The second criterion of treatment ignorability appears similar in spirit to
subset ignorability, but it conditions on the potential outcomes \(M(w)\) and
\(M(b)\) rather than on the actual outcome \(M\). In practice, that distinction
may not be too significant; in theory, however, the difference between the two
is large. As we show in Theorem~\ref{thm:main-body} below, treatment
ignorability alone---even with its strong assumption on the first stage---is not
sufficient to ensure the \(\cdes\) is identified by the second-stage data.

Finally, we consider mediator ignorability and the related mediator monotonicity
condition.

\begin{defn}[Mediator ignorability]
 For \(z \in \{w,b\}\) and \(m \in \{0,1\}\), \emph{mediator ignorability} is
 satisfied when
    \begin{equation}
    \label{eq:mediator_ignorability}
      Y(z,m) \indep M(w) \mid Z = z, M(b) = 1, X.
    \end{equation}
\end{defn}

\begin{defn}[Mediator monotonicity]
  \emph{Mediator monotonicity} is satisfied when
    \begin{equation}
    \label{eq:mediator_monotonicity}
      M(b) \geq M(w).
    \end{equation}
\end{defn}

In our running example, mediator ignorability means that the charging potential
outcomes \(Y(z,m)\) are independent of \emph{one} of the arrest potential
outcomes---\(M(w)\), the arrest decision for (counterfactually) white
individuals---conditional on the observed covariates \(X\), and among
individuals of race \(Z = z\), who would be arrested if they were Black. The
asymmetry in this condition stems from the additional mediator monotonicity
constraint considered by \citeauthor{knox-2019}: \(M(b) \geq M(w)\), meaning
that an individual who would be arrested if white would also be arrested if
Black. The monotonicity condition is perhaps intuitively plausible given our
understanding of racial discrimination, but the conditional independence
assumption of mediator ignorability appears harder to interpret.

Having introduced the key definitions, we now present our main analytic
result, Theorem~\ref{thm:main-body}, which summarizes and formalizes our
discussion of the various ignorability assumptions and their connections to
estimating discrimination. In particular, we show that sequential ignorability
is a strictly stronger assumption than subset ignorability, and recapitulate
(from Theorem~\ref{thm:main}) that subset ignorability is a sufficient condition
for the difference-in-means estimator \(\Delta_n\) to yield consistent estimates
of the \(\cdes\). Further, we show that treatment ignorability is not a
necessary condition for \(\Delta_n\) to yield consistent estimates. We show this
by explicitly constructing examples for which \(\Delta_n
\stackrel{\text{a.s.}}{\rightarrow} \cdes\), but which violate the treatment
ignorability condition. We additionally show that treatment ignorability is not
a sufficient condition to guarantee consistency, despite its formal resemblance
to the (sufficient) subset ignorability condition. To do so, we construct a
family of observationally equivalent examples that satisfy treatment
ignorability but which have different values of the \(\cdes\). Accordingly, no
estimator, including \(\Delta_n\), can yield a consistent estimate of the
\(\cdes\) for every instance in the family. Importantly, the more conventional
assumption of subset ignorability is sufficient to ensure the \(\cdes\) can be
identified from data on the second-stage decisions.

\begin{thm}
\label{thm:main-body}
  Assume overlap holds, meaning that \(\Pr(Z = z \mid X = x, M = 1) > 0\) for
  all \(x\) and \(z\). Then we have the following collection of implications
  and non-implications:
    \begin{equation*}
      \noindent\makebox[\textwidth]{%
        \begin{tikzpicture}[
          style = {font = \upshape, align = center},
          scale = 4,
          >= to,
          line width = 2pt,
          execute at begin node= \setlength{\baselineskip}{1em},
          impl/.style  = {-{Latex[length=3mm,width=5mm]}},
          nimpl/.style = {-{Latex[length=3mm,width=5mm]}, color = black!40,
              decoration = {markings, mark=at position 0.5 with {
                  \draw [red, line width = 1pt, -]
                  ++ (-4pt, -4pt)
                  -- (4pt, 4pt);}},
              postaction = decorate
          }
        ]
          \node[text width=3cm] (SI) at (.1, 1)
            {Sequential\\Ignorability};
          \node[text width=3cm] (CI) at (1.4, 1)
            {Subset\\Ignorability};
          \node[text width=5cm] (O)  at (2.9, 1)
            {\(\Delta_n\) is a consistent\\estimator of the \(\cdes\)};
          \node[text width=3cm] (TI) at (2.9, 1.6)
            {Treatment\\Ignorability};
          \node[text width=5cm] (TI+) at (2.9, .35)
            {%
              Treatment Ignorability,\\Mediator Ignorability, and\\Mediator
              Monotonicity%
            };

          \draw[impl]  (SI.10) to (CI.170);
          \draw[nimpl] (CI.188) to (SI.352);

          \draw[impl]  (CI.9)  to (O.174);
          \draw[nimpl] (O.186)  to (CI.350);

          \draw[impl] (TI+.114) to (O.238);
          \draw[nimpl] (O.297) to (TI+.70);

          \draw[nimpl] (O.121)  to (TI.238);
          \draw[nimpl] (TI.296)  to (O.63);
        \end{tikzpicture}
      }
    \end{equation*}
\end{thm}

\begin{proof}

Theorem~\ref{thm:main} shows that subset ignorability implies that \(\Delta_n\)
is a consistent estimator of the \(\cdes\). We show the remaining seven
implications and non-implications in turn, starting with the claim that
sequential ignorability implies subset ignorability. In particular, we prove
that the conjunction of treatment ignorability, mediator ignorability, and
mediator monotonicity implies that \(\Delta_n\) is a consistent estimator of the
\(\cdes\)---a fact initially suggested by \citeauthor{knox-2019}

\begin{cas}[Sequential ignorability implies subset ignorability]
\label{cas:SIimpCI}
  The first condition of sequential ignorability, in Eq.~\eqref{eq:si1},  states
  that \(Y(z, m)\) and \(M(z')\) are jointly independent of \(Z\) given \(X\):
  \(\{Y(z, m), M(z')\} \indep Z \mid X\). From this, it immediately follows
  that \(Y(z, m)\) alone is independent of \(Z\) given \(X\): \(Y(z, m) \indep Z
  \mid X\). Now, because \(Y(z, m) \indep M \mid Z, X\)---which is the second
  condition of sequential ignorability, in Eq.~\eqref{eq:si2}---we have that
  \(Y(z, m) \indep \{Z, M\} \mid X\), by the contraction property of conditional
  independence. Therefore, by the weak-union property,
    \begin{equation}
    \label{eq:SIimpCI}
      Y(z, m) \indep Z \mid M, X.
    \end{equation}
  Subset ignorability now follows, as it is the special case in which \(M = 1\)
  in Eq.~\eqref{eq:SIimpCI}.
\end{cas}

\begin{cas}[Subset ignorability does not imply sequential ignorability]
  \label{cas:CInimpSI} Sequential ignorability is an intuitively stronger
  condition than subset ignorability, as the former requires that \(Z\) is
  independent of the mediator potential outcomes \(M(z)\) given \(X\). Indeed,
  the synthetic example given in Section~\ref{sec:ex} satisfies subset
  ignorability but violates sequential ignorability.

  To formally establish our claim, we construct an even simpler example that
  satisfies subset ignorability but not sequential ignorability. First, suppose
  that \(Y(z, 1) = 1\) and \(Y(z, 0) = 0\), deterministically for
  \(z\in\{b,w\}\). In particular, using the language of our policing and
  prosecution application, everyone who is arrested is charged, regardless of
  race. We further set \(X = 1\), which effectively means that there are no
  contextual variables. Finally, we set
    \begin{align}
      \begin{split}
        \Pr&(Z = z, M(b) = m_b, M(w) = m_w) \\
            & = \Pr(Z = z) \cdot \Pr(M(b) = m_b \mid Z = z) \cdot \Pr(M(w) = m_w
              \mid Z = z),
        \end{split}
      \end{align}
  where \(\Pr(Z = z) = \tfrac12\), \(\Pr(M(z) = 1 \mid Z = w) = \tfrac12\), and
  \(\Pr(M(z) = 1 \mid Z = b) = 1\). Note that \(M = M(Z)\) and \(Y = Y(Z,M)\),
  and so the above description fully defines the joint distribution on all the
  relevant variables.

  Now, because \(Y(z,1)\) = 1, we trivially have that \(Y(z,1) \indep Z \mid
  M\), meaning that subset ignorability is satisfied. But, because \(M(z) \not
  \indep Z\), sequential ignorability is violated.
\end{cas}

\begin{cas}[%
  Consistency of \(\Delta_n\) does not imply subset ignorability holds%
]
\label{cas3}
  At a high level, even if the potential outcomes \(Y(z, 1)\) are not
  independant of \(Z\)---violating subset ignorability---\(\Delta_n\) can
  still be a consistent estimator when there is appropriate cancellation. For
  a concrete illustration of this in the context of our two-stage arrest and
  charging application, consider a simple example in which: (1) there are no
  contextual variables (i.e., \(X = 1\)); (2) the population is evenly split
  across race groups (i.e., \(\Pr(Z = z) = \frac12\)); (3) everyone in the
  population is arrested (i.e., \(M = 1\)); and (4) the prosecutor's
  \emph{potential} decisions depend on an arrestee's \emph{actual} race.
  Specifically, we set \(Y(z, 0) = 0\) and \(Y(z, 1)\) to be a Bernoulli
  random variable distributed as follows:
    \begin{equation}
    \label{eq:yz1-proof}
      \Pr(Y(z, 1) = 1 \mid Z) =
        \begin{cases}
          1 & z = b \land Z = b, \\
          0 & z = w \land Z = b, \\
          \frac12 & Z = w.
        \end{cases}
    \end{equation}
  Because \(Y = Y(Z,M)\), the above relationships completely specify the joint
  distribution of \(Y\), \(Z\), \(M\), and \(X\).

  Subset ignorability is violated in this example since, by
  Eq.~\eqref{eq:yz1-proof}, \(Y(z, 1) \not \indep Z\). (Because \(X\) and \(M\)
  are constant, we need not condition on them when considering the subset
  ignorability criterion.) We further have,
    \begin{align*}
      \cdes
        &= \EE[Y(b, 1) \mid M = 1] - \EE[Y(w, 1) \mid M = 1] \\
        &= \left( \EE[Y(b,1) \mid Z = b] - \EE[Y(w, 1) \mid Z = b] \right) \cdot
            \Pr(Z = b) \nonumber \\
        &\hspace{1cm} + \left( \EE[Y(b,1) \mid Z = w] - \EE[Y(w, 1) \mid Z = w]
            \right) \cdot \Pr(Z = w) \\
        &= \left(1 - 0 \right) \cdot \frac 12 + \left( \frac 12 - \frac 12
            \right) \cdot \frac 12\\
        &= \frac12.
    \end{align*}
  Finally,
    \begin{align*}
      \lim_{n\to\infty} \Delta_n
        &\stackrel{\text{a.s.}}{=} \EE[Y \mid Z = b, M = 1] -
            \EE[Y \mid Z = w, M = 1]\\
        &= 1 - \frac 12\\
        & = \cdes.
    \end{align*}
  Thus, even though subset ignorability is violated in this example,
  \(\Delta_n\) yields a consistent estimate of the \(\cdes\).
\end{cas}

\begin{cas}[%
  Consistency of \(\Delta_n\) does not imply treatment ignorability holds%
]
\label{cas:CDESnimplTI}
  Consider the example described in Case \ref{cas:CInimpSI}. As discussed
  there, subset ignorability is satisfied in that example and so, by
  Theorem~\ref{thm:main}, \(\Delta_n\) is a consistent estimator of the
  \(\cdes\). However, that example does not satisfy treatment ignorability, as
  \(M(z) \not \indep Z\), contrary to Eq.~\eqref{eq:ti1}. (Because \(X\) is
  constant, we need not condition on it when evaluating the treatment
  ignorability criterion.)
\end{cas}

\begin{cas}[%
  Consistency of \(\Delta_n\) does not imply that treatment ignorability,
  mediator ignorability, and mediator monotonicity hold%
]
  This is directly implied by Case \ref{cas:CDESnimplTI}.
\end{cas}

\begin{cas}[%
  Treatment ignorability does not imply \(\Delta_n\) is a consistent estimator
  of the \(\cdes\)%
]
  We show, more generally, that the \(\cdes\) is not identifiable under
  treatment ignorability alone. To do so, we construct a family of
  observationally equivalent examples that satisfy treatment ignorability but
  which have different values of \(\cdes\). As a result, no
  estimator---including \(\Delta_n\)---can consistently estimate the \(\cdes\)
  for every example in this family.

  We construct the family of examples as follows. First, as in the other cases,
  we set \(X = 1\), so that there are effectively no contextual variables, and
  we set \(Y(z, 0) = 0\), meaning that if an individual were not arrested, that
  individual could not be charged. Second, we set \(M(b) = 1\), meaning that
  everyone in the population would be arrested if they were Black. Finally, we
  set
    \begin{align}
    \label{eq:cas5-jd}
      \begin{split}
        \Pr&(Y(z, 1) = y_{zm}, M(w) = m_w, Z = z) \\
          &= \Pr(Y(z, 1) = y_{zm} \mid M(w) = m_w) \cdot \Pr(M(w) = m_w) \cdot
              \Pr(Z = z),
      \end{split}
    \end{align}
  where \(\Pr(Z = z) = \tfrac12\), \(\Pr(M(w) = m_w) = \tfrac 12\), and, for
  \(\alpha \in [0,1]\),
    \begin{equation}
    \label{eq:cas5-yzm}
        \Pr(Y(z, 1) = 1 \mid M(w)) =
          \begin{cases}
            \alpha & M(w) = 0 \land z = w,\\
            1 & \text{otherwise}.\\
          \end{cases}
    \end{equation}
  The examples we construct thus differ only in the choice of \(\alpha\).

  Now, regardless of \(\alpha\), these examples all satisfy treatment
  ignorability. To see this, note that \(M(w) \indep Z\) by
  Eq.~\eqref{eq:cas5-jd} and \(M(b) \indep Z\) since \(M(b)\) is constant.
  Consequently, the first condition of treatment ignorability is satisfied.
  Eq.~\eqref{eq:cas5-jd} further implies that \(Y(z,1) \indep Z \mid M(w)\) and,
  since \(Y(z, 0)\) is constant, \(Y(z, 0) \indep Z \mid M(w)\), establishing
  the second condition of treatment ignorability. (Because \(M(b)\) and \(X\)
  are constant, we need not condition on them when considering the two treatment
  ignorability conditions.)

  We next show that all these examples are observationally equivalent.
  Intuitively, observational equivalence stems from the fact that the only
  difference between the examples is in the distribution of \(Y(w,1)\) for those
  individuals with \(M(w) = 0\). But for those with \(M(w) = 0\), who would not
  be arrested if they were white, we never observe \(Y(w,1)\).

  Now, to rigorously establish observational equivalence, we must show that
  \(\Pr(X = x, Y = y, Z = z \mid M = 1)\) does not depend on the value of
  \(\alpha\). Because \(X\) is constant, we need only consider \(\Pr(Y = y, Z =
  z \mid M = 1)\). First, observe that
    \begin{align*}
      \Pr(M = 1)
        & = \Pr(M(w) = 1, Z = w) + \Pr(M(b) = 1, Z = b) \\
        & = \Pr(M(w) = 1)\cdot \Pr(Z = w) + \Pr(Z = b) \\
        & = \tfrac34.
    \end{align*}
  Further, note that
    \begin{equation*}
      \Pr(Y = y, Z = z, M = 1) = \Pr(Y(z, 1) = y, Z = z, M(z) = 1),
    \end{equation*}
  and consider the case \(z = b\). Then, because \(Y(b,1) = 1\) and
  \(M(b) = 1\),
    \begin{equation}
    \label{eq:cas5-z=b}
      \Pr(Y = y, Z = b, M = 1) =
        \begin{cases}
          0  & y = 0,\\
          \tfrac12 & y = 1.
        \end{cases}
    \end{equation}
  Now consider the case \(z = w\). By Eq.~\eqref{eq:cas5-jd},
    \begin{align*}
        \Pr(Y&(w, 1) = y, Z = w, M(w) = 1) \\
          & = \Pr(Y(w, 1) = y \mid M(w) = 1) \cdot \Pr(M(w) = 1) \cdot
              \Pr(Z = w) \\
          & = \Pr(Y(w, 1) = y \mid M(w) = 1) \cdot \tfrac14.
    \end{align*}
  By Eq.~\eqref{eq:cas5-yzm}, \(\Pr(Y(w, 1) = 1 \mid M(w) = 1)  = 1\), and so,
    \begin{equation}
    \label{eq:cas5-z=w}
      \Pr(Y = y, Z = w, M = 1) =
        \begin{cases}
          0  & y = 0,\\
          \tfrac14 & y = 1.
        \end{cases}
    \end{equation}
  Finally, combining Eqs.~\eqref{eq:cas5-z=b} and \eqref{eq:cas5-z=w} with the
  fact that \(\Pr(M=1) = \tfrac34\), we have
    \begin{equation*}
      \Pr(Y = y, Z = z \mid M = 1) =
        \begin{cases}
          0 & y = 0,\\
          \tfrac23 & y = 1 \land z = b, \\
          \tfrac13 & y = 1 \land z = w.
        \end{cases}
    \end{equation*}
  In particular, \(\Pr(Y = y, Z = z \mid M = 1)\) does not depend on \(\alpha\),
  and so the examples are all observationally equivalent.

  We conclude the proof by showing that the \(\cdes\) differs across these
  examples. First, it remains to calculate \(\Pr(M(w) = m_w \mid M = 1)\). To do
  so, note that
    \begin{align*}
        \Pr(M(w) = 1, M = 1)
            &= \Pr(M(w) = 1, Z = w) + \Pr(M(w) = 1, M(b) = 1, Z = b)\\
            &= \Pr(M(w) = 1) \cdot \Pr(Z = w) + \Pr(M(w) = 1) \cdot \Pr(M(b) =
              1) \cdot \Pr(Z = b)\\
            &= \frac 12 \cdot \frac 12 + \frac 12 \cdot 1 \cdot \frac 12\\
            &= \frac 12,
    \end{align*}
  and so, since \(\Pr(M = 1) = \frac{3}{4}\), it follows that \(\Pr(M(w) = 1
  \mid M = 1) = \tfrac 23\).

  Consequently, we have
    \begin{align*}
       \cdes
          &= \EE[Y(b, 1) \mid M = 1] - \EE[Y(w, 1) \mid M = 1] \\
          &= \Pr(M(w) = 1 \mid M = 1) \\
          &\hspace{1.5cm} \cdot \left( \EE[Y(b,1) \mid M(w) = 1, M = 1] -
              \EE[Y(w, 1) \mid M(w) = 1, M = 1] \right)\\
          &\hspace{1cm} + \Pr(M(w) = 0 \mid M = 1) \\
          &\hspace{1.5cm} \cdot \left( \EE[Y(b,1) \mid M(w) = 0, M = 1] -
              \EE[Y(w, 1) \mid M(w) = 0, M = 1] \right)\\
          &= \frac 23 \cdot (1 - 1) + \frac 13 \cdot (1 - \EE[Y(w, 1) \mid M(w)
            = 0, Z = b])\\
          &=\frac{1 - \alpha}{3},
    \end{align*}
  where second to last equality follows from Eq.~\eqref{eq:cas5-yzm} and the
  fact that the event \(\{M(w) = 0 \land M = 1\}\) equals \(\{M(w) = 0 \land Z =
  b\}\); the final equality also follows from Eq.~\eqref{eq:cas5-yzm}, as well
  as the fact that \(Y(z,1) \indep Z \mid M(w)\). We have thus constructed a
  family of observationally equivalent examples that satisfy treatment
  ignorability but which have different \(\cdes\), implying that the  \(\cdes\)
  is not in general identifiable under treatment ignorability alone.
\end{cas}

\begin{cas}[%
  Treatment, mediator ignorability, and mediator monotonicity jointly imply
  \(\Delta_n\) is a consistent estimator of the \(\cdes\)%
]
  The proof is in two pieces. First, we derive an expression for the \(\cdes\)
  holding \(X\) constant, and then prove the general claim.

  Supposing \(X = x\) is constant, recall that by definition \(M = 1\) if and
  only if \(M(z) = 1\) where \(Z = z\). By mediator monotonicity, \(M(b)
  \geq M(w)\). Therefore, the event \(\{M = 1\}\) can be partitioned into the
  following two events:
    \begin{itemize}
      \item \label{defn:E1} \(E_1 = \{M(b) = 1 \land M(w) = 1\}\),
      \item \label{defn:E2} \(E_2 = \{Z = b \land M(b) = 1 \land M(w) = 0\}\).
    \end{itemize}
  Recall the definition of the \(\cdes\) in Definition~1. It follows from
  the law of total expectation that:
    \begin{align}
      \cdes &= \EE[Y(b, 1) - Y(z, 1) \mid M = 1]\nonumber\\
        \begin{split}
        \label{eq:mti_decomp}
          &= \EE[Y(b, 1) - Y(z, 1) \mid E_1] \cdot \Pr(E_1 \mid M = 1)\\
          &\hspace{1cm} + \EE[Y(b, 1) - Y(z, 1) \mid E_2] \cdot
              \Pr(E_2 \mid M = 1)
        \end{split}
    \end{align}

  Now, we examine each of these summands in  turn. First, consider the \(E_1\)
  term:
    \begin{equation*}
      \EE[Y(b, 1) - Y(w, 1) \mid E_1] = \EE[Y(b, 1) \mid E_1]
        - \EE[Y(w, 1) \mid E_1]
    \end{equation*}
  By the definition of \(E_1 = \{M(b) = 1 \land M(w) = 1\}\) and the second
  treatment ignorability condition, Eq.~\eqref{eq:ti2}, we are free to condition
  both terms on the right hand side by levels of \(Z\), yielding
    \begin{equation}
    \label{eq:e1}
      \EE [Y(b, 1) \mid Z = b, E_1] - \EE [Y(w, 1) \mid Z = w, E_1]
          = \EE [Y \mid Z = b, E_1] - \EE [Y \mid Z = w, E_1],
    \end{equation}
  where equality follows from replacing potential outcomes by their realized
  values according to the definition of \(Y = Y(M, Z)\).

  Next, consider the \(E_2\) term. Again,
    \begin{align*}
      \EE[Y(b, 1) - Y(w, 1) \mid E_2]
        &= \EE[Y(b, 1) \mid E_2] - \EE[Y(w, 1) \mid E_2].
    \end{align*}
  It follows from mediator ignorability, Eq.~\eqref{eq:mediator_ignorability},
  and the definition of \(E_2\) that
    \begin{align*}
      \EE[Y(w, 1) \mid E_2]
        &= \EE[Y(w, 1) \mid Z = b, M(b) = 1, M(w) = 0]\\
        &= \EE[Y(w, 1) \mid Z = b, M(b) = 1, M(w) = 1]\\
        &= \EE[Y(w, 1) \mid Z = w, M(b) = 1, M(w) = 1],
    \end{align*}
  where the last equality follows from treatment ignorability,
  Eq.~\eqref{eq:ti2}. Replacing potential outcomes with their realizations, it
  follows that
    \begin{equation}
    \label{eq:e2}
      \EE[Y(b, 1) - Y(w, 1) \mid E_2] = \EE [Y \mid Z = b, E_2] -
          \EE [Y \mid Z = w, E_1].
    \end{equation}

  Now, we substitute Eqs.~\eqref{eq:e1}~and~\eqref{eq:e2} into
  Eq.~\eqref{eq:mti_decomp}.
    \begin{align}
      \cdes
        &= \left( \EE[Y \mid Z = b, E_1] - \EE[Y \mid Z = w, E_1] \right) \cdot
            \Pr (E_1 \mid M = 1)\nonumber\\
        &\hspace{1cm} + \left( \EE[Y \mid Z = b, E_2] - \EE[Y \mid Z = w, E_1]
            \right) \cdot \Pr (E_2 \mid M = 1)\nonumber\\
        &= \left( \EE[Y \mid Z = b, E_1] - \EE[Y \mid Z = w, M = 1] \right)
            \cdot \Pr (E_1 \mid M = 1)\nonumber\\
        &\hspace{1cm} + \left( \EE[Y \mid Z = b, E_2] - \EE[Y \mid Z = w, M = 1]
            \right) \cdot \Pr (E_2 \mid M = 1)\nonumber\\
        &= \big( \EE[Y \mid Z = b, E_1] \cdot \Pr(E_1 \mid M = 1) + \EE[Y \mid Z
            = b, E_2] \cdot \Pr(E_2 \mid M = 1) \big) \nonumber\\
        &\hspace{1cm} - \big( \EE [Y \mid Z = w, M = 1] \cdot \left( \Pr (E_1
            \mid M = 1) + \Pr (E_2 \mid M = 1) \right) \big) \nonumber\\
        &= \EE [Y \mid Z = b, M = 1] - \EE [Y \mid Z = w, M = 1]
            \label{eq:mti_final},
    \end{align}
  where the second equality follows from the fact that \(\{M = 1 \land Z = w\} =
  \{E_1 \land Z = w\}\) by mediator monotonicity, and the last equality follows
  from the facts that \(\{M = 1 \land Z = b \land E_1\} = \{Z = b \land E_1\}\),
  \(\{M = 1 \land Z = b \land E_2\} = \{Z = b \land E_2\}\), and \(\Pr(E_1 \mid
  M = 1) + \Pr(E_2 \mid M = 1) = 1\).

  Now, suppose that \(X\) is not constant. Conditioning \(Y\), \(Z\), and \(M\)
  on \(X = x\), it follows from the law of total expectation that
    \begin{align}
      \EE [Y(b, 1) - Y(w, 1) \mid M = 1]
        &= \sum_x \EE [Y(b, 1) - Y(w, 1) \mid M = 1, X = x] \cdot \Pr (X = x
            \mid M = 1) \nonumber\\
          \begin{split}
          \label{eq:generalize}
            &= \sum_x \EE [Y \mid Z = b, M = 1, X = x] \cdot \Pr(X = x \mid M =
                1)\\
            &\hspace{1cm} - \EE [Y \mid Z = w, M = 1, X = x] \cdot \Pr(X = x
                \mid M = 1),
          \end{split}
    \end{align}
  where the second equality follows from Eq.~\eqref{eq:mti_final}, using the
  fact that \(X\) is constant on each of the events \(\{ X = x \}\).
  Eq.~\eqref{eq:generalize} is identical to the expression in the statement of
  Theorem~\ref{thm:main}, and so the estimator \(\Delta_n\) converges almost
  surely to the quantity on the right-hand side of Eq.~\eqref{eq:generalize} by
  precisely the same argument as there.
\end{cas}

\end{proof}

\section{Analysis of a Restricted Family of Distributions}
\label{sec:counterexamples}

Theorem~\ref{thm:main-body} shows that treatment ignorability, mediator
ignorability, and mediator monotonicity are jointly sufficient but not necessary
to identify the \(\cdes\) from data on second-stage decisions. We show that this
non-necessity holds even if we restrict to distributions compatible with a
particular causal DAG considered by \citeauthor{knox-2019}, shown in
Figure~\ref{fig:klm-dag}, where an unobserved confounder \(Q\) directly
influences the first-stage decisions \(M\) (e.g., arrests) and the second-stage
decisions \(Y\) (e.g., charging). To do so, we explicitly construct a
counterexample in which: (1) the joint distribution of random variables is
compatible with this causal DAG; (2) mediator ignorability is violated; and (3)
subset ignorability is satisfied, which in turn implies that the stratified
difference-in-means \(\Delta_n\) is a consistent estimator of the \(\cdes\), by
Theorem~\ref{thm:main}.

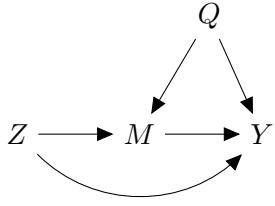
\begin{figure}[t]
  \centering
    \begin{center}
    \begin{tikzpicture}
      \node                                     (race) at (0,0) {\(Z\)};
      \node[right = of race]                    (arrest)        {\(M\)};
      \node[right = of arrest]                  (charge)        {\(Y\)};
      \node[above right = 1 and 0.3 of arrest]  (confound)      {\(Q\)};

      \draw[->] (race)      to                      (arrest);
      \draw[->] (arrest)    to                      (charge);
      \draw[->] (race)      to[out = 315, in = 225] (charge);
      \draw[->] (confound)  to                      (arrest);
      \draw[->] (confound)  to                      (charge);
    \end{tikzpicture}
  \end{center}
  \caption{\emph{
      A causal DAG considered by \citeauthor{knox-2019} In the context of our
      charging example, \(Z\) indicates race, \(M\) indicates arrest decisions,
      \(Y\) indicates charging decisions, and \(Q\) is an unobserved confounder.
      Even when one restricts to distributions compatible with this DAG, the
      \citeauthor{knox-2019} conditions are not necessary to
      non-parametrically identify the \(\cdes\) from data on second-stage
      decisions.
  }}
\label{fig:klm-dag}
\end{figure}

\begin{prop}
  There exists a structural causal model (SCM) compatible with the causal DAG in
  Figure~\ref{fig:klm-dag} which violates mediator ignorability but satisfies
  subset ignorability.
\end{prop}

\begin{proof}
  We start by explicitly constructing an SCM that is (faithfully) compatible
  with the DAG in Figure~\ref{fig:klm-dag}. Our SCM has the following
  independent exogenous variables:
    \begin{align*}
      U_Z & \sim \unif(\{w,b\}), \\
      U_Q & \sim \unif(\{1,2,3,4\}), \\
      U_M & \sim \unif((0,1)), \\
      U_Y & \sim \unif((0,1)),
    \end{align*}
  where \(U_Z\) and \(U_Q\) are uniformly distributed over the specified
  discrete sets, and \(U_M\) and \(U_Y\) are uniform over the unit interval.
  Now, the structural equations are given by:
    \begin{align*}
      f_Z(u_z) &= u_z, \\
      f_Q(u_q) &= u_q, \\
      f_M(z, q, u_m)
        &= \B 1 \left(u_m \leq (1 + \B 1(z = b)) \cdot \frac {\B 1(q=1) + \B
          1(z = b \land q = 3) + \B 1(z = w \land q = 2)} {2} \right),\\
      f_Y(z, m, q, u_y)
        &= m \cdot \B 1 \left(u_y \leq (1 + \B 1(z = b)) \cdot \frac {\B 1(q =
          1)}{2} \right),
    \end{align*}
  where \(\B 1\) denotes the indicator function and \(\land\) denotes
  conjunction (i.e., the \texttt{and} operator). For avoidance of doubt, \(Z =
  f_Z(U_Z)\), \(Q = f_Q(U_Q)\), \(M = f_M(Z, Q, U_M)\), and \(Y = f_Y(Z, M, Q,
  U_Y)\). Further, the potential arrest outcomes are given by \(M(z) = f_M(z, Q,
  U_M)\), and the bivariate potential charge outcomes are given by \(Y(z,m) =
  f_Y(z, m, Q, U_Y)\).

  \paragraph{Mediator ignorability is violated.}

  First, note that \(Z \indep \{Y(b,1), M(w), M(b)\}\), because \(Z\) is a
  function of \(U_Z\), and \(\{Y(b,1), M(w), M(b)\}\) are functions of \(U_Y\),
  \(U_Q\), and \(U_M\), which are jointly independent of \(U_Z\). Now, applying
  this fact and conditioning on \(Q\), we have that,
    \begin{align*}
      \Pr(Y(b,1) = 1 &\mid M(w) = m_w, M(b) = 1, Z = z) \\
        &= \Pr(Y(b,1) = 1 \mid M(w) = m_w, M(b) = 1) \\
        &= \sum_{q=1}^4 \Pr(Y(b,1) = 1 \mid M(w) = m_w, M(b) = 1, Q =
          q) \\
        &\hspace{1.5cm}\cdot \Pr(Q = q \mid M(w) = m_w, M(b) = 1)\\
        &= \sum_{q=1}^4 \Pr(f_Y(b,1,q,U_Y) = 1 \mid M(w) = m_w, M(b)
          = 1, Q = q) \\
        &\hspace{1.5cm} \cdot \Pr(Q = q \mid M(w) = m_w, M(b) = 1).
    \end{align*}
  Next, observe that \(f_Y(b,1,q,U_Y) = \B 1(q=1)\),
  and so
    \begin{align}
      & \Pr(Y(b,1) = 1 \mid M(w) = m_w, M(b) = 1, Z = z) \nonumber \\
      & \hspace{1cm} =  \Pr(Q = 1 \mid M(w) = m_w, M(b) = 1) \nonumber \\
      & \hspace{1cm} = \frac{\Pr(M(w) = m_w, M(b) = 1 \mid Q = 1) \cdot
        \Pr(Q=1)}{\sum_{q=1}^4 \Pr(M(w) = m_w, M(b) = 1 \mid Q = q) \cdot
        \Pr(Q=q)} \nonumber \\
      & \hspace{1cm} = \frac{\Pr(M(w) = m_w, M(b) = 1\mid Q = 1)}{\sum_{q=1}^4
        \Pr(M(w) = m_w, M(b) = 1 \mid Q = q)}. \label{eq:bayes}
    \end{align}
  The second equality above follows from Bayes' rule, and the third follows from
  the fact that \(\Pr(Q=q) = 1/4\).

  Finally, we compute \(\Pr(M(w) = m_w, M(b) = 1 \mid Q = q)\). Note that
    \begin{align}
      M(w) &= f_M(w, Q, U_M) \nonumber \\
      &= \B 1 \left( U_M \leq \frac{\B 1(Q=1) + \B 1(Q=2)}{2} \right)\nonumber\\
      & = \label{eq:M_w}
        \begin{cases}
          \B 1(U_M \leq 1/2) & Q \in \{1,2\}, \\
          0 & \text{otherwise}. \\
        \end{cases}
    \end{align}
  Likewise,
    \begin{align}
      M(b) &= f_M(b, Q, U_M) \nonumber \\
        &= \B 1 (U_M \leq (\B 1(Q=1) + \B 1(Q=3))) \nonumber \\
        & = \label{eq:M_b}
          \begin{cases}
            1 & Q \in \{1,3\}, \\
            0 & \text{otherwise}. \\
          \end{cases}
    \end{align}
  As a result,
    \begin{align*}
      \Pr(M(w) = m_w, M(b) = 1 \mid Q = q) &=
        \begin{cases}
          1/2 & q=1, \\
          1 & q=3 \land m_w = 0, \\
          0 & \text{otherwise}. \\
        \end{cases}
    \end{align*}
  Thus, by Eq.~\eqref{eq:bayes},
    \begin{equation*}
      \Pr(Y(b, 1) = 1 \mid M(w) = 1, M(b) = 1, Z = z) = 1,
    \end{equation*}
  while
    \begin{equation*}
      \Pr(Y(b, 1) = 1 \mid M(w) = 0, M(b) = 1, Z = z) = \frac 13.
    \end{equation*}
  Therefore, \(Y(b, 1) \not\indep M(w) \mid M(b) = 1, Z = z\), meaning that
  mediator ignorability does not hold.

  \paragraph{Subset ignorability holds.}

  Similar to the above, we have that
    \begin{align}
      Y(b,1) &= f_Y(b, 1, Q, U_Y) \nonumber \\
        &= \B 1 ( U_Y \leq\B 1(Q=1)) \nonumber \\
        & = \label{eq:Y_b_1}
          \begin{cases}
            1 & Q =1 \\
            0 & \text{otherwise}, \\
          \end{cases}
    \end{align}
  and
    \begin{align}
      Y(w,1) &= f_Y(w, 1, Q, U_Y) \nonumber \\
        &= \B 1 ( U_Y \leq\B 1(Q=1)/2) \nonumber \\
        & = \label{eq:Y_w_1}
          \begin{cases}
            \B 1(U_Y \leq 1/2) & Q =1 \\
            0 & \text{otherwise}. \\
          \end{cases}
    \end{align}

  Now, as before, \(Z \indep \{Q, M(w), M(b)\}\), since \(Z\) is a function of
  \(U_Z\), and \(\{Q, M(w), M(b)\}\) are functions of \(U_Q\) and \(U_M\), which
  are jointly independent of \(U_Z\). Consequently,
    \begin{align}
      \Pr(Q = 1 \mid M = 1, Z = z)
        &= \Pr(Q=1 \mid M(z) = 1, Z = z) \nonumber \\
        &= \Pr(Q=1 \mid M(z) = 1) \nonumber \\
        &= \frac{\Pr(M(z) = 1 \mid Q = 1) \cdot \Pr(Q=1)} {\sum_{q=1}^4
          \Pr(M(z) = 1 \mid Q = q) \cdot \Pr(Q=q)} \nonumber \\
        &= \frac{1}{2}, \label{eq:Q}
    \end{align}
  where the last equality follows from Eqs.~\eqref{eq:M_w} and \eqref{eq:M_b},
  together with the fact that \(\Pr(Q = q) = 1/4\), and that \(U_M \indep Q\).

  Finally, conditioning on \(Q\), we have
    \begin{align*}
      & \Pr(Y(b,1) = 1 \mid M=1, Z =z) \\
      & \hspace{1cm} = \sum_{q=1}^4 \Pr(Y(b,1) = 1 \mid M=1, Z=z, Q=q)
        \cdot \Pr(Q = q \mid M=1, Z = z) \\
      & \hspace{1cm} = \Pr(Q = 1 \mid M=1, Z = z)  \\
      & \hspace{1cm} = \frac{1}{2},
    \end{align*}
  where the second equality follows from Eq.~\eqref{eq:Y_b_1},
  and the third from Eq.~\eqref{eq:Q}. Similarly,
    \begin{align*}
      & \Pr(Y(w,1) = 1 \mid M=1, Z =z) \\
      & \hspace{1cm} = \sum_{q=1}^4 \Pr(Y(w,1) = 1 \mid M=1, Z=z, Q=q)
        \cdot \Pr(Q = q \mid M=1, Z = z) \\
      & \hspace{1cm} = \Pr(U_Y \leq 1/2 \mid M=1, Z=z, Q=1) \cdot \Pr(Q = 1 \mid
        M=1, Z = z)  \\
      & \hspace{1cm} = \Pr(U_Y \leq 1/2) \cdot \Pr(Q = 1 \mid M=1, Z = z) \\
      & \hspace{1cm} = \frac{1}{4},
    \end{align*}
  where the second equality follows from Eq.~\eqref{eq:Y_w_1}, the third from
  the fact that \(U_Y \indep \{M, Z, Q\}\), and the fourth from
  Eq.~\eqref{eq:Q}. Therefore, \(\Pr(Y(b, 1) = y \mid M = 1, Z = b) = \Pr(Y(b,
  1) \mid M = 1, Z = w)\) and similarly for \(Y(w, 1)\). In particular, this
  means that \(Y(z,1) \indep Z \mid M = 1\), and so subset ignorability holds.
\end{proof}

\section{Extending Theorem~\ref{thm:main} to Allow for Continuous Covariates}
\label{app:cont}

Theorem~\ref{thm:main} in the main text shows that subset
ignorability---together with overlap---implies the \(\cdes\) is
nonparametrically identified, where, for simplicity, we proved the result for
discrete covariates \(X\). We now extend that result to allow for continuous
covariates. At a conceptual level, the extension is straightforward: we first
condition on \(X\), then appeal to subset ignorability to condition on \(Z\),
and, finally, use consistency to replace potential outcomes by their observed
values. In the general case, however, typically \(\Pr(X = x) = 0\), and so one
must take care to define expressions that nominally condition on these
probability-zero events.

Recall that in the discrete case, the primary conditional expectations, treated
as functions of \(z\) and \(x\), are of the form
  \begin{align}
    \EE [Y | Z = z, X = x, M = 1]
      & = \sum_y y \frac{\Pr(Y = y, Z = z, X = x \mid M = 1)}{\Pr(Z = z, X = x
        \mid M = 1)} \nonumber\\
      & = \sum_y y \frac{\Pr(Y = y, Z = z, X = x \mid M = 1)}{\Pr(Z = z \mid X =
        x, M = 1) \Pr(X = x \mid M = 1)}. \label{eq:condexp}
  \end{align}
Overlap ensures that the denominator in \eqref{eq:condexp} is non-zero, and,
accordingly, that the conditional expectation is well-defined. In the continuous
case, to address conditioning on probability-zero events, conditional
probabilities are defined as random variables rather than simple numeric
quantities (cf. \citet{billingsley2008probability}). Further, if the random
variables \(\Pr (Z = z \mid X, M = 1) > 0\) a.s.\ for \(z \in \{w, b\}\)---a
condition that we call generalized overlap---then the expression \(\EE [Y \mid Z
= z, X = x, M = 1]\) is a well-defined function of \(z\) and \(x\), as in the
discrete case, up to a set of measure zero with respect to the pushforward
measure \(\mu_{X \mid M = 1}\) for each fixed \(z\).\footnote{%
  The pushforward measure \(\mu_{X \mid M = 1}\) is the measure on \(\C
  X\)---the range of \(X\)---given by \(\mu_{X \mid M = 1}[A] = \Pr(X \in A \mid
  M = 1)\) for measurable \(A \subseteq \C X\).
}\textsuperscript{,}\footnote{%
  To see this, first note that, in general, \(\EE [Y \mid Z = z, X = x, M = 1]\)
  is uniquely defined up to a set of measure zero with respect to the
  pushforward measure \(\mu_{Z, X \mid M = 1}\). Now, for fixed \(z\), suppose,
  toward a contradiction, that \(f_1(x)\) and \(f_2(x)\) are two versions of
  \(\EE [Y \mid Z = z, X = x, M = 1]\) that differ on a set \(A\) such that
  \(\Pr(X \in A \mid M = 1) > 0\). Then, by the generalized overlap condition,
  \(\Pr(Z = z, X \in A \mid M = 1) = \int_A \Pr(Z = z \mid X = x, M=1) \,
  \text{d} F_{X \mid M = 1} > 0\), contradicting the fact that \(f_1(x) \neq
  f_2(x)\) only on a null set with respect to the pushforward measure \(\mu_{Z,
  X \mid M = 1}\).
}

We now state and prove the extension of Theorem~\ref{thm:main}, with the
understanding that the conditional probabilities and expectations below are
defined according to the usual measure-theoretic conventions.

\begin{thm}
  Suppose \(Y(z,1)\), \(Z\), \(M\), and \(X\) satisfy subset ignorability, and
  that generalized overlap holds---i.e., for \(z \in \{b, w\}\), \(\Pr (Z = z
  \mid X, M = 1) > 0\) a.s. Then, the \(\cdes\) equals
    \begin{align}
      \begin{split}
      \label{eq:thmmaincont}
       & \int_{\C X} \EE[Y \mid Z = b, X=x, M=1] \, \textnormal{d} F_{X \mid M =
         1} \\
       & \hspace{1cm} - \int_{\C X} \EE[Y \mid Z = w, X=x, M=1] \,
         \textnormal{d} F_{X \mid M = 1},
      \end{split}
    \end{align}
    where \(\C X\) denotes the range of \(X\) and \(\textnormal{d} F_{X \mid M =
    1}\) denotes integration over \(\C X\) with respect to the pushforward
    measure \(\mu_{X \mid M = 1}\).
\end{thm}

\begin{proof}
  By conditioning on \(X\), we have,
    \begin{align}
      \cdes
        &= \mathbb{E}[Y(b,1) - Y(w,1) \mid M = 1]\nonumber \\
        &= \int_{\C X} \mathbb{E}[Y(b,1) - Y(w,1) \mid X = x, M = 1] \, \text{d}
          F_{X \mid M = 1} \nonumber \\
        &= \int_{\C X} \mathbb{E}[Y(b,1) \mid X = x, M = 1] - \mathbb{E}[Y(w,1)
          \mid X = x, M = 1] \, \text{d} F_{X \mid M = 1}. \label{eq:coneq0}
    \end{align}
  Now, subset ignorability gives that
    \begin{equation}
    \label{eq:coneq1}
        \mathbb{E}[Y(z,1) \mid X = x, M = 1] = \mathbb{E}[Y(z,1) \mid X = x,
          Z=z, M = 1]\ \text{a.s.},
    \end{equation}
  where generalized overlap ensures that the right-hand side of
  Eq.~\eqref{eq:coneq1} is well-defined up to a set of measure zero with respect
  to \(\text{d} F_{X \mid M = 1}\). Substituting this expression into
  Eq.~\eqref{eq:coneq0}, and then appealing to consistency to replace potential
  outcomes with their observed values, we have
    \begin{align*}
      \cdes
        &= \int_{\C X} \mathbb{E}[Y(b,1) \mid X = x, Z = b, M = 1] -
          \mathbb{E}[Y(w,1) \mid X = x, Z = w, M = 1] \, \text{d} F_{X \mid M =
          1} \\
        &= \int_{\C X} \mathbb{E}[Y(Z,M) \mid X = x, Z = b, M = 1] -
          \mathbb{E}[Y(Z,M) \mid X = x, Z = w, M = 1] \, \text{d}F_{X \mid M =
          1} \\
        &= \int_{\C X} \mathbb{E}[Y \mid X = x, Z = b, M = 1] - \mathbb{E}[Y\mid
          X = x, Z = w, M = 1] \, \text{d}F_{X \mid M = 1}.
    \end{align*}
\end{proof}

All of the quantities in Eq.~\eqref{eq:thmmaincont} (i.e., the distribution of
\(X\) and the conditional expectations) are functions of observables,
establishing that the \(\cdes\) is identified by data on second-stage decisions.
One may adopt a variety of approaches to estimate the terms in
Eq.~\eqref{eq:thmmaincont}, including model-based strategies, as we do in
Section~\ref{sec:empirical}. One may also adopt non-parametric estimation
strategies, wherein continuous covariates are appropriately binned into discrete
sets. For further treatment of these issues, see, for example,
\cite{gelman2013bayesian}, \cite{friedman2001elements}, and
\cite{tsybakov2008introduction}.

\section{Summary Statistics of Prosecution Dataset}
\label{app:demographics}

We present summary statistics, disaggregated by demographic group, of the
dataset used to conduct the empirical analysis of prosecutorial charging
decisions in Section~\ref{sec:empirical}.

\begin{table}[t]
  \caption{\emph{%
    Breakdown of charges, prior arrests, prior convictions, and weapons
    involvement among individuals arrested for a felony offense between 2013 and
    2019 in a major U.S.\ county, as analyzed in Section~\ref{sec:empirical}.
    Note that only Black and non-Hispanic white individuals are analyzed in our
    race-based analysis, so within-row sums differ between race and gender.
  }}
  \label{tbl:summary}
  \begin{center}
    \begin{tabular}{rlrrrr}
      \hline
      {} & {} & Female & Male & Black & White \\
      \hline
      \textbf{Charge:}
         & arson                       & 50    & 218   & 79    & 124   \\
      {} & assault                     & 2,062 & 9,160 & 4,131 & 3,662 \\
      {} & burglary                    & 664   & 3,334 & 1,762 & 1,535 \\
      {} & burglary (auto)             & 125   & 1,216 & 737   & 338   \\
      {} & driving under the influence & 26    & 90    & 17    & 50    \\
      {} & drug-related offense        & 736   & 5,643 & 2,815 & 1,881 \\
      {} & forcible rape               & 3     & 112   & 49    & 26    \\
      {} & forgery                     & 60    & 187   & 103   & 84    \\
      {} & hit-and-run                 & 10    & 42    & 18    & 17    \\
      {} & kidnapping                  & 4     & 24    & 12    & 1     \\
      {} & manslaughter (vehicular)    & 2     & 8     & 4     & 3     \\
      {} & motor vehicle theft         & 83    & 323   & 161   & 141   \\
      {} & other felony offense        & 287   & 2,115 & 877   & 910   \\
      {} & other sex offense           & 8     & 180   & 85    & 54    \\
      {} & robbery                     & 458   & 2,186 & 1,536 & 589   \\
      {} & theft                       & 877   & 3,385 & 1,723 & 1,539 \\
      {} & weapons offense             & 80    & 882   & 513   & 218   \\
      {} & willful homicide            & 53    & 178   & 65    & 60    \\
      \\[-7pt]
      \textbf{Prior arrests:}
         & 0   & 3,192 & 13,491 & 5,881 & 5,682 \\
      {} & 1   & 1,005 & 5,486  & 2,972 & 1,977 \\
      {} & 2   & 553   & 3,315  & 1,929 & 1,134 \\
      {} & 3   & 321   & 2,199  & 1,215 & 786   \\
      {} & 4   & 191   & 1,558  & 895   & 540   \\
      {} & 5   & 116   & 1,066  & 591   & 364   \\
      {} & 6   & 78    & 697    & 381   & 242   \\
      {} & 7   & 40    & 469    & 248   & 157   \\
      {} & 8   & 29    & 316    & 180   & 103   \\
      {} & 9   & 18    & 204    & 111   & 84    \\
      {} & 10+ & 45    & 482    & 283   & 163   \\
      \\[-7pt]
      \textbf{Prior convictions:}
         & 0  & 4,451 & 20,775 & 9,937 & 8,290 \\
      {} & 1  & 748  & 5,221   & 2,961 & 1,773 \\
      {} & 2  & 278  & 2,055   & 1,111 & 730   \\
      {} & 3  & 78   & 766     & 429   & 261   \\
      {} & 4  & 23   & 312     & 171   & 116   \\
      {} & 5+ & 10   & 154     & 77    &  62   \\
      \\[-7pt]
      \textbf{Weapons:}
         & no weapon involved & 4,442 & 24,932 & 12,810 & 9,434 \\
      {} & weapon involved    & 1,146 & 4,351  & 1,876  & 1,798 \\
      \hline
    \end{tabular}
  \end{center}
\end{table}

\begin{figure}[b]
  \begin{center}
    \begin{subfigure}{0.48\textwidth}
      \begin{center}
        \includegraphics[height=2in]{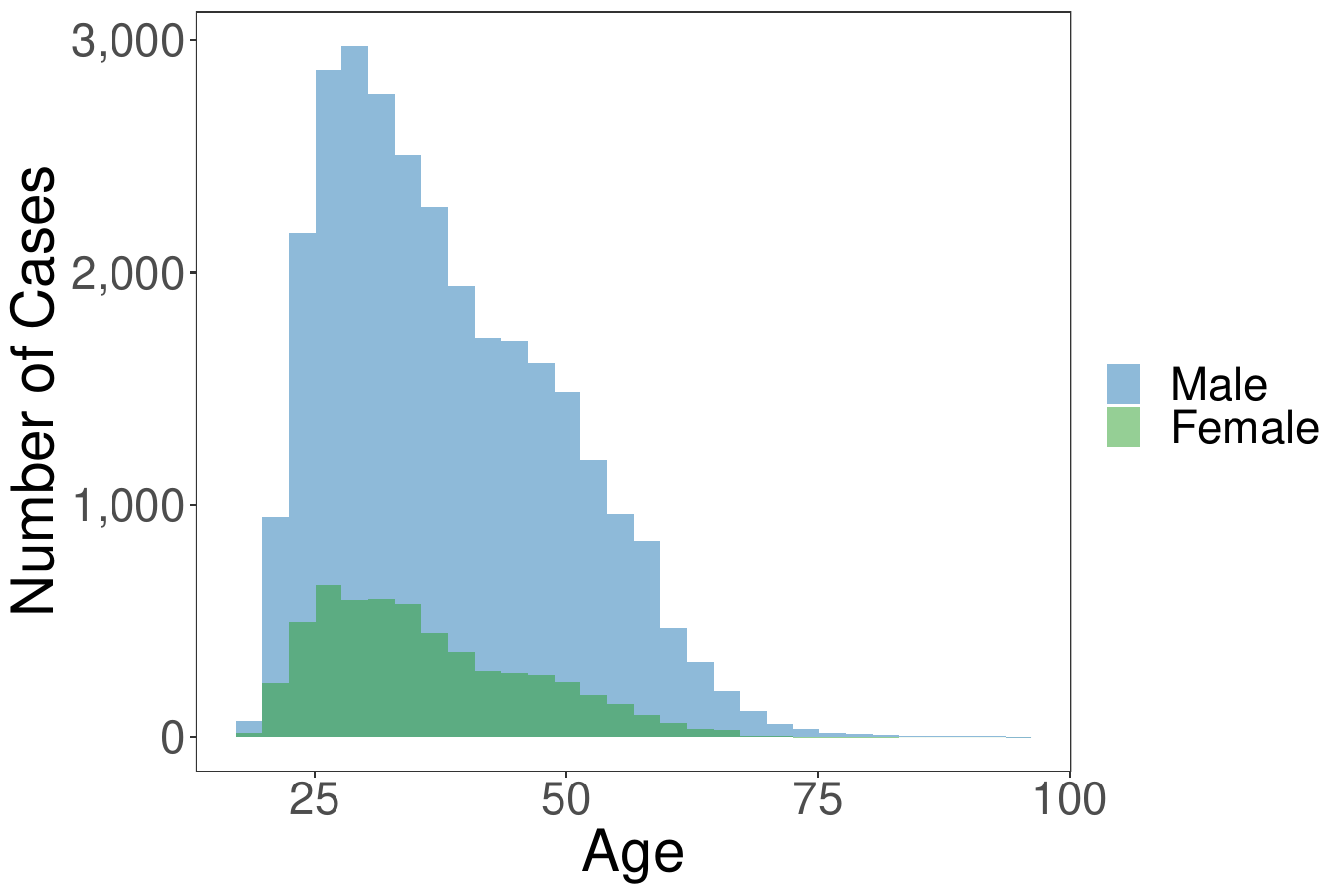}
      \end{center}
      \subcaption{Distribution of age by gender}
    \end{subfigure}
    \begin{subfigure}{0.48\textwidth}
      \begin{center}
        \includegraphics[height=2 in]{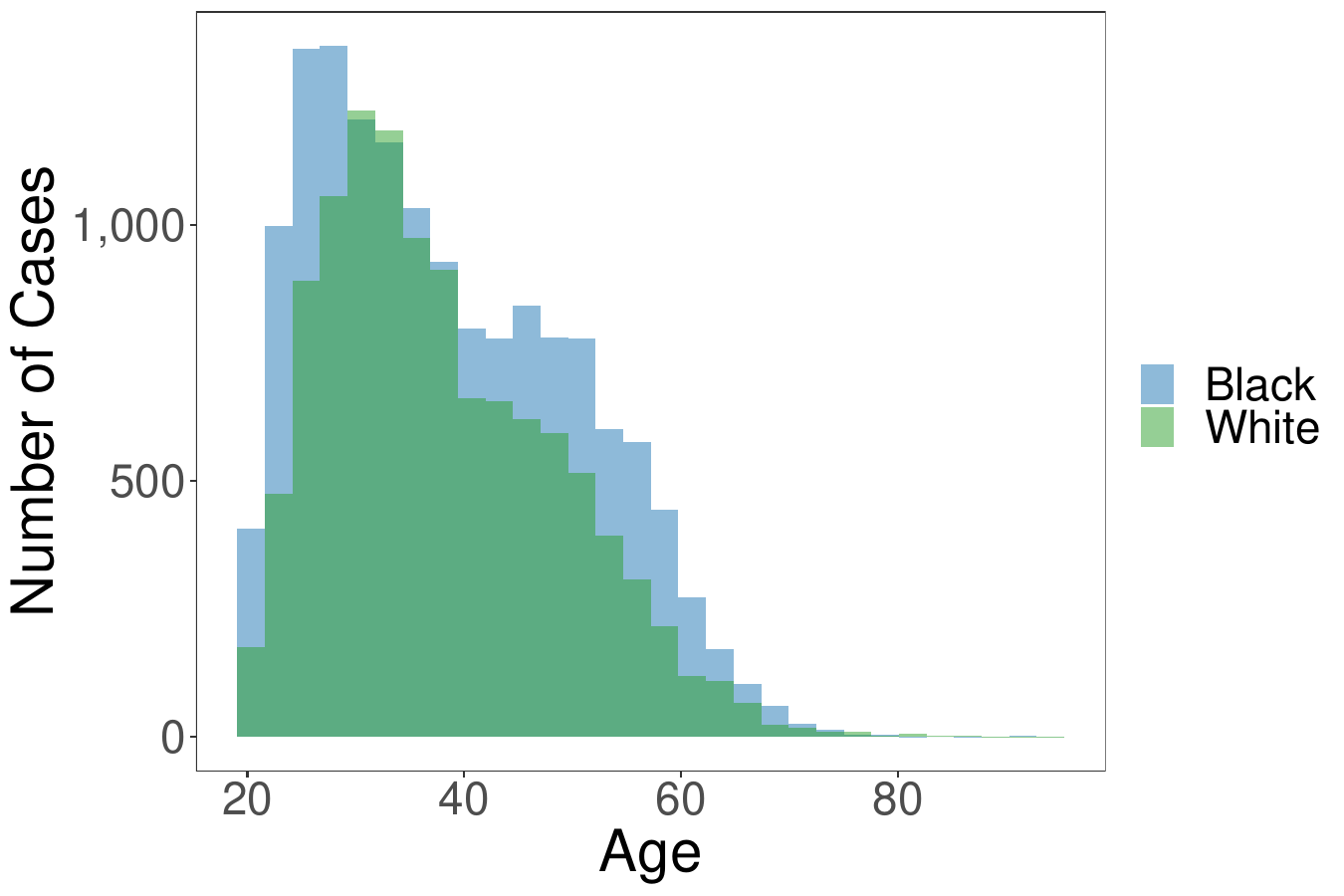}
      \end{center}
      \subcaption{Distribution of age by race}
    \end{subfigure}
  \end{center}
  \caption{\emph{%
    Breakdown of age by gender and by race of individuals arrested for a felony
    offense between 2013 and 2019 in a major U.S.\ county, as analyzed in
    Section~\ref{sec:empirical}.
  }}
\end{figure}

\end{document}